\documentclass[conference]{IEEEtran}

\usepackage{cite}
\usepackage{amsmath,amssymb,amsfonts}
\usepackage{algorithm}
\usepackage{algpseudocode}
\usepackage{graphicx}
\usepackage{textcomp}
\usepackage{xcolor}
\usepackage{url}
\usepackage{amsthm}
\usepackage{subcaption}
\theoremstyle{definition}
\newtheorem{dfn}{Definition}

\newtheorem{lem}[dfn]{Lemma}
\newtheorem{thm}[dfn]{Theorem}

\def\BibTeX{{\rm B\kern-.05em{\sc i\kern-.025em b}\kern-.08em
    T\kern-.1667em\lower.7ex\hbox{E}\kern-.125emX}}

\begin{document}
%
\title{Tie-Breaking Rule Based on\\ Partial Proof of Work in a Blockchain}

\author{\IEEEauthorblockN{Akira Sakurai}
\IEEEauthorblockA{\textit{Kyoto University} \\
Tokyo, Japan }
\and
\IEEEauthorblockN{Kazuyuki Shudo}
\IEEEauthorblockA{\textit{Kyoto University} \\
Kyoto, Japan }
}



\maketitle

\begin{abstract}
In the area of blockchain, numerous methods have been proposed for suppressing intentional forks by attackers more effectively than the random rule. However, all of them, except for the random rule, require major updates, rely on a trusted third party, or assume strong synchrony. Hence, it is challenging to apply these methods to existing systems such as Bitcoin.

To address these issues, we propose another countermeasure that can be easily applied to existing proof of work blockchain systems. Our method is a tie-breaking rule that uses partial proof of work, which does not function as a block, as a time standard with finer granularity. By using the characteristic of partial proof of work, the proposed method enables miners to choose the last-generated block in a chain tie, which suppresses intentional forks by attackers.
Only weak synchrony, which is already met by existing systems such as Bitcoin, is required for effective functioning.

We evaluated the proposed method through a detailed analysis that is lacking in existing works. In networks that adopt our method, the proportion of the attacker hashrate necessary for selfish mining was approximately 0.31479 or higher, regardless of the block propagation capability of the attacker. Furthermore, we demonstrated through extended selfish mining that the impact of Match against pre-generated block, which is a concern in all last-generated rules, can be mitigated with appropriate parameter settings.
\end{abstract}

\section{Introduction}
Blockchain is the underlying technology in numerous decentralized currency systems, including Bitcoin\cite{bitcoin}. Unlike Byzantine fault-tolerant distributed consensus algorithms such as PBFT\cite{PBFT}, blockchain allows any node to participate in the system. This enables the elimination of trusted third parties during transaction processing. In this study, we focus specifically on proof of work (PoW) blockchains.

In blockchain systems, transactions are processed in units known as blocks. Each block consists of a header and a list of transactions that are processed in that block. A block is valid when its header's hash value is less than a certain target value. The process of identifying such a header is known as mining. The nodes that participate in the system and generate blocks are referred to as miners. When miners successfully generate a block, they publish it and the published block is shared across the network. Miners can receive rewards for the blocks that they successfully generate.

Each block contains information that references the previous block. The Genesis Block, as defined by the protocol, is eventually reached by recursively following these references. The sequence of blocks that is obtained by following these references is uniquely determined for each block, and this sequence is referred to as a chain. Miners use a fork choice rule based on the blocks that they possess to determine the main chain. For example, in Bitcoin, the longest chain is selected (which is known as the longest-chain rule). Each miner engages in mining on the main chain.

Situations exist in which the fork choice rule alone cannot uniquely determine the main chain, which is known as a chain tie. For example, this occurs when multiple longest chains exist in Bitcoin. The rules for uniquely determining a chain during a chain tie are known as tie-breaking rules. In Bitcoin, the tie-breaking rule is to choose the chain that is first seen (hereafter referred to as the first-seen rule).

The security of a blockchain system depends on each miner adhering to the aforementioned protocol. Therefore, it is desirable that following the protocol is the most profitable action for each miner. However, Eyal et al. introduced a new mining strategy known as selfish mining (SM) and demonstrated that blockchain does not always satisfy the necessary incentive compatibility\cite{majorityisnotenough}. Specifically, they demonstrated that miners can gain more rewards by withholding a successfully generated block instead of publishing it immediately, and then releasing it when another node publishes a block.

Eyal et al. proposed a tie-breaking rule that randomly selects a chain during chain ties (hereafter referred to as the random rule), thereby introducing the approach of modifying tie-breaking rules as a countermeasure to SM \cite{majorityisnotenough}.

To suppress intentional forks more effectively than the random rule, several studies have been conducted\cite{majorityisnotenough, oneweirdtricktostopselfishminers, PreventingBitcoinSelfishMiningUsingTransactionCreationTime, PublishorPerish, ZeroBlock, CounteringSelfishMininginBlockchains, PreventingSelfishMininginPublicBlockchainUsing, TFTstrategy}. However, their solutions face the following three significant challenges: they necessitate major updates, require a trusted third party, or assume strong synchrony among miners.

To the best of our knowledge, Heilman was the first to propose a tie-breaking rule that selects the latest chain (hereafter referred to as last-generated rules)\cite{oneweirdtricktostopselfishminers}. Last-generated rules enable each miner to distinguish between chains that are withheld by attackers and others. Consequently, last-generated rules can suppress intentional forks by attackers more effectively than the random rule.

Heilman introduced a separate time standard from the blockchain to construct the last-generated rule. However, it requires a trusted third party, which makes it difficult to implement in trustless systems such as Bitcoin. Besides this method, several other approaches have also been proposed. However, apart from the random rule, they require major system updates that could undermine the system, necessitate a trusted third party, or demand strong synchrony, which is not feasible in actual systems.

We propose a last-generated rule to address these challenges. We use partial PoW as a new time standard for constructing our method. A partial PoW is data that contain a weaker PoW compared to the original block. In the case of Bitcoin, a partial PoW is a block header that meets a certain level of PoW. Blockchain represents the flow of time through the accumulation of blocks. As partial PoWs can be generated more easily and in greater numbers than blocks, they can be used to represent a more granular flow of time than the blockchain.

The proposed method is practical and applicable to existing systems that employ PoW blockchains, including Bitcoin.
Our proposed method requires only that there exists a widely known upper bound on the message propagation times, and there exists a widely known upper bound on the absolute values of clock drifts. This synchrony is the minimum requirement for last-generated rules.
Our approach does not require a specific trusted third party. Furthermore, it does not require updates that can destabilize the system, such as soft forks or hard forks.

We conduct a detailed analysis of the proposed method. We further refine the classification of intentional chain ties that are initiated by attackers into two categories: Match against post-generated block and Match against pre-generated block. We thoroughly evaluate our method for each category. Additionally, we introduce extended selfish mining (ESM), which considers Match against pre-generated block, specifically to analyze concerns in the last-generated rule. This mining strategy allows for a more detailed evaluation of last-generated rules, including the proposed method\cite{oneweirdtricktostopselfishminers, PreventingBitcoinSelfishMiningUsingTransactionCreationTime, PreventingSelfishMininginPublicBlockchainUsing}. This demonstrates the significance of parameter settings in the last-generated rule.

\textbf{Organizations.} The remainder of this paper is organized as follows. Section \ref{intentionalfork} describes intentional forks by attackers. Section \ref{relatedworks} discusses related works. Section \ref{proposedmethod} describes the proposed method. Section \ref{evaluation} demonstrates the effectiveness of the proposed method. Section \ref{detail} presents a detailed discussion, primarily focusing on the application of the proposed method and consideration of clock drift. Section \ref{conclusion} concludes the paper.

\section{Intentional Forks by Attackers}
\label{intentionalfork}
In a blockchain, a fork refers to a situation in which a chain splits into multiple branches. There are two types of forks: those caused intentionally by miners not following the protocol and those that occur naturally. Attacks such as SM\cite{majorityisnotenough, CounteringSelfishMininginBlockchains, optimalSelfishMining, stubbornMining, Forkafterwithholding} are typical examples of intentional fork exploitation. Among intentional forks, chain ties that are intentionally created by attackers are known as intentional chain ties. SM and similar attacks exploit these intentional chain ties. The proposed method is a tie-breaking rule that aims to suppress intentional chain ties. Suppressing intentional chain ties leads to suppressing intentional forks, which in turn helps to mitigate attacks such as SM.

We consider SM as an example to demonstrate the exploitation of intentional forks. First, the proportion of a miner's own blocks in the main chain is defined as the relative revenue. The goal of SM is to achieve a higher relative revenue than the proportion of the attacker's hashrate in the total network. If this is accomplished, the attack is considered successful. The attacker withholds the block that they have successfully mined instead of publishing it immediately. When an honest miner generates a block, the attacker publishes their block, creating a chain tie. If a block is generated that extends the attacker's chain after the chain tie, the honest miner’s block is invalidated, and vice versa. If the attacker mines blocks consecutively, they publish their block when the length difference between their chain and the honest miner’s chain is one. In this scenario, the honest miner's chain is invalidated because the attacker's chain is longer. Thus, the attacker can invalidate the blocks that are created by honest miners and increase their relative revenue by manipulating the timing of publishing their mined blocks in a manner that does not follow the protocol.

Intentional chain ties by attackers can be further divided into two types: Match against post-generated block and Match against pre-generated block. The former occurs when an attacker waits for an honest miner to generate and publish their block after successfully generating a block, and then simultaneously publishes their own pre-generated block, causing a chain tie. The latter occurs when the attacker continues to mine on a chain without updating it after an honest miner generates a block. Subsequently, the attacker publishes its own newly generated block, which leads to a chain tie.

Match against pre-generated block is difficult under the first-seen rule, however, it becomes feasible under last-generated rules. Last-generated rules impose a time limit on accepting chains in a tie to address this issue. This time limit is referred to as the acceptance window $w$. Without such a limit, attackers can always initiate chain ties against honest miners and invalidate their blocks. Conversely, Match against pre-generated block can be mitigated to an extent by setting this limit. We detail the extent to which setting an acceptance window quantitatively suppresses Match against pre-generated block in Section \ref{matchagainstpre-generatedblock}.

\section{Related Work} \label{relatedworks}
Eyal et al. introduced SM\cite{majorityisnotenough}. They also demonstrated that when SM attackers initiate chain ties under the first-seen rule, the proportion of honest miners who mine to extend the attacker's chain is strongly dependent on the attacker’s block propagation ability. Under the first-seen rule, if the attacker can deliver their block to all honest miners earlier than the blocks that honest miners generate, the honest miners' blocks can be reliably invalidated. Therefore, Eyal et al. proposed a tie-breaking rule for randomly selecting the main chain from the chains that are involved in the chain tie. This approach enables the proportion of honest miners mining to extend the attacker's chain during a chain tie to be reduced to 0.5, irrespective of the attacker's block-propagation ability. This tie-breaking rule had been adopted by Ethereum~\cite{ethereum}.

Heilman was the first to propose a last-generated rule as a countermeasure against SM\cite{oneweirdtricktostopselfishminers}. This method involves embedding the timestamps that are provided by a trusted third party outside the system into blocks. These timestamps are compared during a chain tie to determine the block that was generated more recently. Blocks that are generated by attackers and by honest miners can be distinguished using this method. However, the provision of timestamps relies on a trusted third party, which makes it difficult to implement in real blockchain systems. In addition, potential drawbacks exist, such as the risk of the trusted third party supplying fraudulent timestamps to attackers.

Lee et al. proposed a last-generated rule that is distinct from the approach of Heilman\cite{PreventingBitcoinSelfishMiningUsingTransactionCreationTime}. In this method, timestamps are first embedded into transactions at the time of creation. During a chain tie, the average timestamps of the transactions that are contained in each block are compared to differentiate between blocks generated by attackers and honest miners. One vulnerability of this method is that it overlooks the possibility of attackers creating transactions with more recent timestamps and including them in their blocks, thereby freely manipulating the average timestamps of transactions in the block. Furthermore, each transaction being required to embed its creation timestamp can negatively affect the transaction-processing capacity of the system.

All of the aforementioned methods are tie-breaking rules. Zhang et al. proposed a fork choice rule known as weighted FRP to suppress SM\cite{PublishorPerish}. This method involves including uncle blocks in a block to consider the number of published blocks in the main chain selection. However, this approach requires changes to the fork choice rule, which may be expensive to implement and can result in multiple forks. In addition, this method may lead to prolonged forks in situations in which the network is fragmented. 

Zhang et al. also proposed another fork choice rule against SM\cite{in-block}. The method uses in-blocks similar to partial PoWs.
However, the method needs a soft fork and suppresses only a specific behavior of an attacker. That limits the effectiveness of the method.

Solat et al. proposed a method named ZeroBlock in which empty blocks are generated when no block is produced for a certain period, thereby invalidating withheld blocks\cite{ZeroBlock}. However, this method requires precise clock synchronization across the entire system, which is a stringent requirement, especially for global-scale decentralized systems such as Bitcoin. Moreover, this approach assumes that blocks in PoW blockchains are generated at regular intervals, but the interval between block generations follows an exponential distribution in actual systems, which would limit the effectiveness of their proposed method. Additionally, implementing this method requires significant updates, such as hard forks, which makes it challenging to apply to systems that are already operational (e.g., Bitcoin).

Saad et al. proposed a fork choice rule with a similar concept to that of Lee et al.\cite{CounteringSelfishMininginBlockchains}. Their method uses high accuracy in predicting the block height, which is the time at which transactions are likely to be included in a block. Similar to the method of Lee et al., the average predicted processing times of transactions that are included in a block are used to distinguish between blocks generated by attackers and honest miners. However, their method also does not address the possibility of attackers including arbitrary transactions in their blocks. Furthermore, it applies only to specific patterns of attacker behavior, which limits its effectiveness.

Reno et al. proposed a method that combines the ZeroBlock approach suggested by Solat et al. with a tie-breaking rule that compares timestamps in blocks\cite{PreventingSelfishMininginPublicBlockchainUsing}. However, attackers can also arbitrarily determine the timestamps that are included in blocks in this method. Furthermore, the problems associated with ZeroBlock remain, such as the assumption of clock synchronization and the need for incompatible updates during implementation.

Sun et al. proposed a method that allows honest miners to delay block propagation to a selfish miner\cite{TFTstrategy}. The method assumes that the network adopts the first-seen rule. The method reduces the winning rate of the attacker by delaying notification of the new block generated by honest miners. However, the method overlooks the following two critical issues. Firstly, they assume that all honest miners adopt the method. Even if one of the honest miners does not adopt the method, the effectiveness of the method is reduced. This is because honest miners who do not adopt the method may send blocks to the attacker as usual. Secondly, they assume that the blockchain network consists solely of mining pools. In reality, numerous solo miners are also participating in the network. The attacker can exploit this by setting up nodes to facilitate block propagation to itself.

Although the first-seen rule is primarily used in the current Bitcoin, discussions on alternative tie-breaking rules are ongoing \cite{cunicula2013why}. Several tie-breaking rules have been proposed that allow all miners to decide unanimously on a main chain in the case of a chain tie. One proposal is a method that selects the chain to which the most computing resources are devoted. Initially, each miner publishes a block header that achieves a certain difficulty level. These headers do not need to function as valid headers for actual blocks. In the event of a chain tie, a chain whose head is referenced by numerous block headers is selected. Our method also involves the sharing of block headers that have achieved a certain difficulty level. A major difference between our method and the aforementioned approach is that our method is characterized by its ability to select the latest chain. In addition, our method considers the number of block headers that are referenced by blocks belonging to a chain, rather than the number of headers that reference blocks belonging to a chain.

\begin{figure}[tb]
\begin{center}
\includegraphics[width=0.5\textwidth]{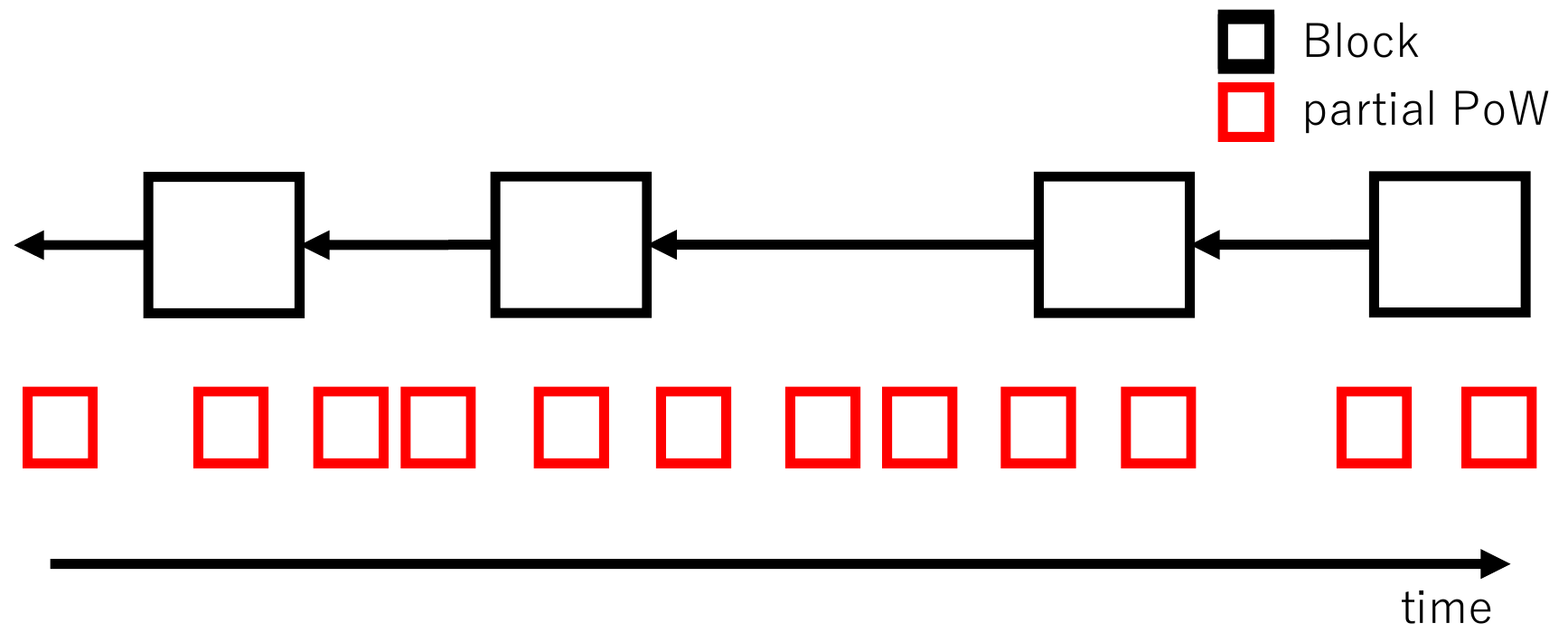}
\end{center}
\caption{Blockchain functions as a clock that progresses block by block. Partial PoWs are easier to generate than blocks. Hence, the number of partial PoWs is larger than that of blocks. Consequently, the blockchain functions as a more precise clock by using partial PoWs.}
\label{timeflow}
\end{figure}

\section{Proposed Method} \label{proposedmethod}
Our proposed method is a type of last-generated rule. We use partial PoW, which originally had no functional use as a block, as a new time standard. Partial PoWs are easier to generate, can be produced in larger quantities than in blocks, and thus can function as a more granular clock (Figure \ref{timeflow}).

The proposed method is practical. Widespread sharing of partial PoWs across the system is required for our method to be effective. However, a partial PoW typically has a smaller size than blocks and does not consume significant network resources from miners. In addition, the implementation of our method has a negligible impact on the transaction-processing capacity of the system. Our method does not require updates that could divide the hashrate of the system, such as hard forks or soft forks, nor does it require a trusted third party. Although a certain level of synchrony in the message propagation times and the progression of miners' clocks (clock drift) are necessary, these requirements are achievable in most blockchain systems and do not compromise the high level of asynchrony that is inherent in PoW blockchains. These factors make our method more practical than previous methods.

Due to its practicality, the proposed method can easily be applied to existing systems such as Bitcoin. In addition, the proposed method is a tie-breaking rule, which provides the following two advantages. Firstly, the proposed method does not generate stale blocks resulting from its partial application. Secondly, the proposed method can be applied to existing fork choice rules. For example, it can be applied to rules such as GHOST\cite{GHOST}, which had been partially adopted in Ethereum\cite{ethereum}, and weighted FRP\cite{PublishorPerish}.

\subsection{Our Model}
Three assumptions regarding synchrony are required for the effectiveness of the proposed method. The first assumption is the existence of a widely known value $\Delta_B$ such that the 100th percentile propagation time of any block is less than $\Delta_B$. Similarly, a widely known value $\Delta_P$ should exist that ensures that the 100th percentile propagation time of any partial PoW is less than $\Delta_P$. The first assumption seems viable based on measurements of Bitcoin~\cite{statsofbitcoin}. Considering that a partial PoW is significantly smaller than a block, which affects the message propagation time, the second assumption also appears to be reasonable.

The final assumption is that the impact of clock drift among miners is sufficiently small to be negligible. For simplicity, We adopt this assumption in our explanation. However, this assumption can easily be replaced with a weaker one: there exists a widely known value $D$ such that the absolute value of the clock drift of any miner is less than $D$. We elaborate on how clock drift affects the proposed method in Section \ref{drift}. In Bitcoin, the PoW blockchain ensures functionality provided that the clock discrepancies of miners do not exceed one hour. Our method can be integrated into the system without compromising the high level of asynchrony that is inherent in PoW blockchains. This is because we make assumptions only about clock drift, not about clock offset.

The above assumptions that there exists a widely known upper bound on the message propagation times, and there exists a widely known upper bound on the absolute values of clock drifts are the minimum requirement for any last-generated rule. This is because last-generated rules have an acceptance window and need to measure how much time has passed since the blocks arrived. 

We make several assumptions regarding attackers in our model. First, attackers are assumed to be capable of instantly becoming aware of any message that is publicly released on the network. In addition, they are assumed to have the ability to deliver any message to any miner instantaneously. Our assumption allows attackers to detect and propagate blocks in any way.

\begin{figure}[tb]
\begin{center}
\includegraphics[width=0.5\textwidth]{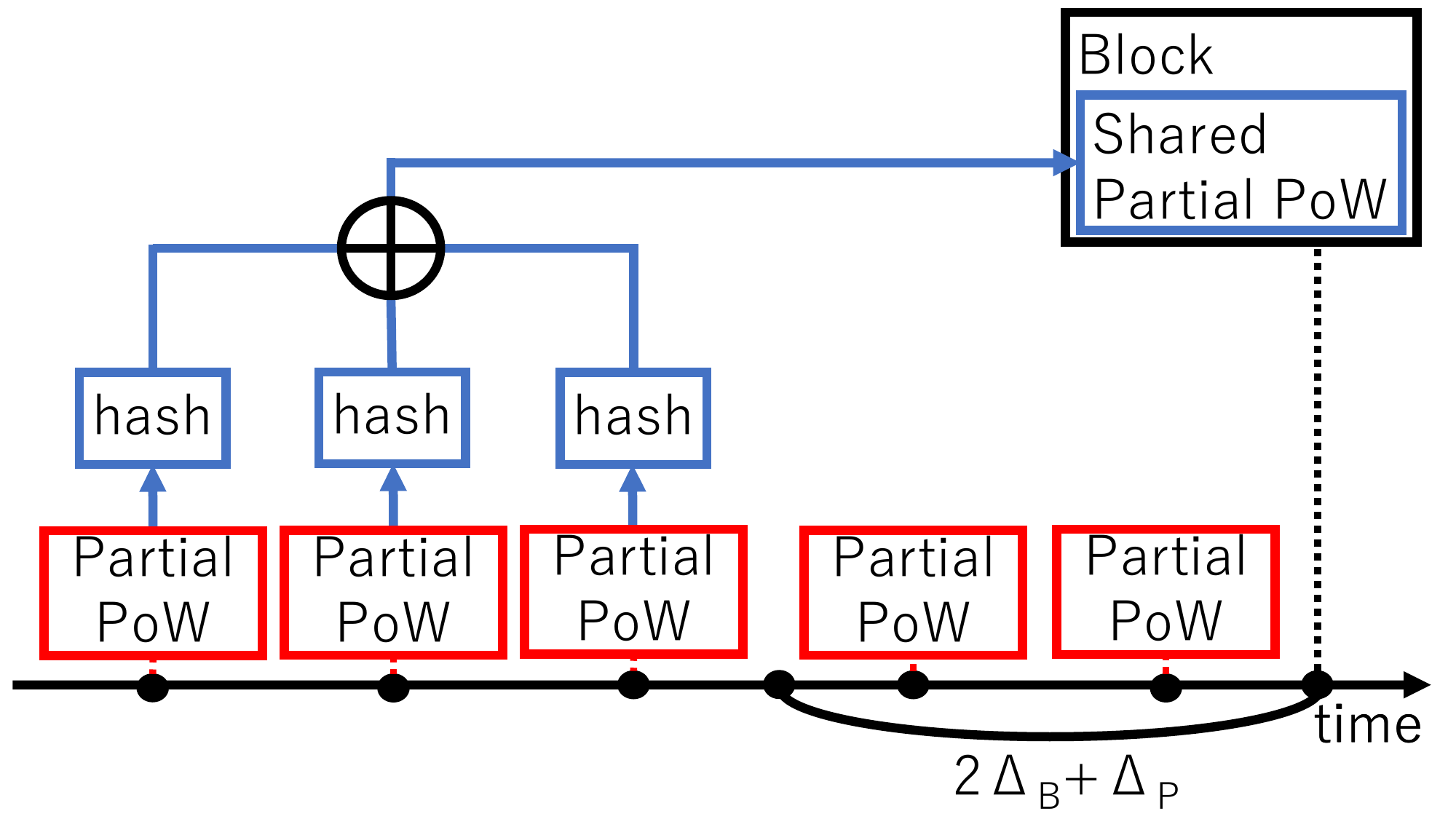}
\end{center}
\caption{How to include shared partial PoW in a block. Each miner considers a partial PoW ``sufficiently shared" if $2 \Delta_B + \Delta_P$ has elapsed since its arrival. Each miner includes sufficiently shared partial PoWs into a block and mines.}
\label{shared}
\end{figure}

\begin{table}[tb]
\centering
\caption{List of variables.}
\begin{tabular}{ll}
\hline
Variable &  \begin{tabular}{l} Definition\end{tabular} \\
\hline
$h$ & \begin{tabular}{l} Block header\end{tabular}\\
\textit{Target} &\begin{tabular}{l}Target value in mining \end{tabular} \\
$H()$ & \begin{tabular}{l} Cryptographic hash function  \end{tabular} \\
$n$ & \begin{tabular}{l}Difficulty adjuster\end{tabular}  \\
$T$ & \begin{tabular}{l} Average block generation interval \end{tabular}  \\
$V$ & \begin{tabular}{l} Set of miners\end{tabular} \\
$N_i$ & \begin{tabular}{l} Total hash computations by miners who do not include a\\ partial PoW generated by miner $i$ into shared partial PoWs\end{tabular} \\
$p$ & \begin{tabular}{l} Probability of 
successfully generating a block \\in one hash computation\end{tabular} \\
$M_i$ & \begin{tabular}{l} Hashrate of miner $i$ \end{tabular}\\
$M_{all}$ & \begin{tabular}{l}Total hashrate of the system \end{tabular} \\
$\alpha$ & \begin{tabular}{l}Proportion of hash rate occupied by the attacker\end{tabular}  \\
$\Delta_B$ & \begin{tabular}{l}Upper bound of block propagation time \end{tabular}  \\
$\Delta_P$ &\begin{tabular}{l} Upper bound of partial PoW propagation time\end{tabular}  \\
$T_W$ & \begin{tabular}{l} Hashrate-weighted average block propagation time\end{tabular} \\
$\gamma$ & \begin{tabular}{l} Proportion of honest miners mining on the attacker's chain\\during match against post-generated block \end{tabular} \\
$\gamma'$ & \begin{tabular}{l} Proportion of honest miners mining on the attacker's chain\\during match against pre-generated block \end{tabular} \\
$s$ & \begin{tabular}{l} Unresponsive time  \end{tabular}\\
$o$ & \begin{tabular}{l} Probability of block generation from the network \\during unresponsive time \end{tabular} \\
$d_i$ & \begin{tabular}{l} Clock drift of miner $i$  \end{tabular}\\
$D$ & \begin{tabular}{l} Upper bound of absolute value of clock drift  \end{tabular}\\
$x_m$ & \begin{tabular}{l}Measured value of time progression \end{tabular}  \\
$x_r$ & \begin{tabular}{l}Actual value of time progression \end{tabular}  \\

\hline
\end{tabular}
\end{table}

\subsection{Details}
Here is the method with the fixed difficulty adjuster (hereafter referred to as the n-fixed method). We elaborate on the method with the variable difficulty adjuster (hereafter referred to as the n-variable method) in Section \ref{applicationToSystems}.

Partial PoW is a block header with a generation difficulty that is $1/n$ times lower than that of a valid block, where $n$ is the difficulty adjuster. That is, a block header $h$ functions as a valid block header when it satisfies $H(h) < \textit{Target}$, where $H$ is a cryptographic hash function and \textit{Target} is a value that is determined by the difficulty. However, it also functions as a partial PoW if it satisfies $H(h) < n\textit{Target}$. Upon successfully generating partial PoW, miners publish and share it on the network, regardless of its validity as a block. As opposed to the propagation of transactions, the risk of DoS attacks owing to the propagation of partial PoWs is not a concern. This is because the validity of a partial PoW is verified, thereby providing proof of work\cite{PricingviaProcessingorCombattingJunkMail}.

The proposed method requires each miner to manage partial PoWs independent of the blockchain. It necessitates embedding new information, termed shared partial PoW, into blocks. We elaborate on how to embed shared partial PoW into a block in the case of Bitcoin in Section \ref{applicationToSystems}. Shared partial PoW is defined as the exclusive logical XOR of the hash values of the partial PoWs that belong to the set of valid partial PoWs for which $2 \Delta_B + \Delta_P$ has elapsed since their arrival (see Figure \ref{shared}). In addition, when propagating a block, miners must send the block as well as the set corresponding to the shared partial PoW. Hereafter, we refer to the exclusive logical sum of partial PoWs in the set, the set corresponding to the shared partial PoW, and partial PoWs that are valid and for which $2\Delta_B + \Delta_P$ has elapsed since their arrival as shared partial PoW.

The fork choice rule that is applied in the proposed method is described as follows: Initially, each miner manages the arrival time of the blocks. The arrival time of a chain is determined by the arrival time of the most recently generated block in the chain; that is, the head of the chain. If no chain tie exists, the main chain is uniquely determined by a predefined fork choice rule. 
In the event of a chain tie, the chains in the tie that have arrived within the acceptance window of $\Delta_B$ from the arrival time of the earliest chain in the tie are selected as the main chain candidates. Next, the blocks constituting each chain candidate are extracted. Then, the shared partial PoW of each block is retrieved, and it is verified whether all partial PoWs belonging to the shared partial PoW are sufficiently shared. Here, each miner considers partial PoWs as ``sufficiently shared" if they are valid and $2\Delta_B$ has elapsed since their arrival. If even one partial PoW in a chain is not sufficiently shared, that chain is assigned a value of $-1$. If all partial PoWs are sufficiently shared, the size of the set corresponding to the shared partial PoW of the chain is assigned. Finally, the main chain with the largest assigned value among the candidates is selected. In the event of a tie, a chain is randomly selected from among the tied chains.


Algorithm \ref{algotiebreakingrulebasedonpartial PoW} is an algorithm that applies our proposed method to the fork choice rule. The getMainChain function (line 1) uses multiple chains as arguments and determines one main chain. In line 2, the fork choice rule (forkChoiceRule function) is used to select the main chain candidates. Subsequently, the earliest arrival time ($earliestTime$) among the main chain candidates is determined up to line 6. Next, from the main chain candidates, those whose arrival time is within $w$ from $earliestTime$ are selected (line 9), and among these, a main chain is uniquely determined using the proposed method (line 10). The chainMax function (line 16) uses two chains as arguments and returns the one that contains more shared partial PoWs. Line 17 determines whether $\textit{C}_\textit{1}$ is an empty chain. Thereafter, the chainWeight function (line 31) determines which chain contains more shared partial PoWs and returns the chain with more shared partial PoWs. If the number is equal, one chain is randomly selected. The chainWeight function (line 31) uses a single chain as an argument and returns the number of shared partial PoWs in that chain. First, each block belonging to the chain is selected to obtain the total shared partial PoW of the entire chain (lines 32-35). Subsequently, each shared partial PoW belonging to the chain is selected. If it is valid and sufficiently shared (line 38), 1 is added to \textit{res} (line 39); otherwise, $-1$ is returned (line 41). The validity of the shared partial PoW was verified using the valid function. Finally, \textit{res} is returned (line 44).

\begin{algorithm}
\caption{fork choice rule using proposed method}
\label{algotiebreakingrulebasedonpartial PoW}
\begin{algorithmic}[1]
\Procedure{getMainChain}{$\textit{C}_\textit{1}, \ldots, \textit{C}_\textit{k}$}
    \State $\textit{candidates} \gets \text{forkChoiceRule}(\textit{C}_\textit{1}, \ldots, \textit{C}_\textit{k})$
    
    \State $\textit{earliestTime} \gets 0$
    \For{$\textit{C} \in \textit{candidates}$}
        \State $\textit{earliestTime} \gets \text{min}(\textit{earliestTime}, \textit{C}.\textit{arrivalTime})$
    \EndFor

    \State $\textit{res} \gets \epsilon$
    \For{$\textit{C} \in \textit{candidates}$}
        \If{$\textit{C}.\textit{arrivalTime} - \textit{earliestTime} \leq \textit{w}$}
            \State $\textit{res} \gets \text{chainMax}(\textit{res}, \textit{C})$
        \EndIf
    \EndFor
    \State \Return $\textit{res}$
\EndProcedure \\

\Procedure{chainMax}{$\textit{C}_\textit{1}, \textit{C}_\textit{2}$}
    \If{$\textit{C}_\textit{2} = \epsilon$}
        \State \Return $\textit{C}_\textit{1}$
    \EndIf
    \If{$\text{chainWeight}(\textit{C}_\textit{1}) > \text{chainWeight}(\textit{C}_\textit{2})$}
        \State \Return $\textit{C}_\textit{1}$
    \Else
        \If{$\text{chainWeight}(\textit{C}_\textit{1}) < \text{chainWeight}(\textit{C}_\textit{2})$}
            \State \Return $\textit{C}_\textit{2}$
        \Else
            \State \Return $\textit{C}_\textit{1}$ or $\textit{C}_\textit{2}$ randomly
        \EndIf
    \EndIf
\EndProcedure \\

\Procedure{chainWeight}{$\textit{C}$}
    \State $\textit{SharedPartialPoW} \gets \emptyset$
    \For{$\textit{block} \in \textit{C}$}
        \State $\textit{SharedPartialPoW} = \textit{SharedPartialPoW} \cup \textit{block}.\textit{SharedPartialPoW}$
    \EndFor
    \State $\textit{res} \gets 0$
    \For{$\textit{PartialPoW} \in \textit{SharedPartialPoW}$}
        \If{$\text{valid}(\textit{PartialPoW})$ \textbf{and} $\textit{currentTime} - \textit{PartialPoW}.\textit{arrivalTime} > 2\Delta_B$}
            \State $\textit{res} = \textit{res} + 1$
        \Else
            \State \Return $-1$
        \EndIf
    \EndFor
    \State \Return $\textit{res}$
\EndProcedure 
\end{algorithmic}
\end{algorithm}

\subsection{Parameter Settings}\label{parameter}
We examine the parameter settings and several associated properties. 

First, we discuss the parameter settings for the acceptance window of the chains in a tie.
\begin{thm}
    When an attacker performs a Match against post-generated block, any honest miner can receive the block that is generated by another honest miner before their own acceptance window ends if the acceptance window $w$ is at least $\Delta_B$.
\end{thm}
\begin{proof}
    Recall the assumption that an attacker can instantly detect a block that is generated by an honest miner and immediately deliver the block that they have generated to all miners without delay. This implies that the earliest start time for the acceptance window of any honest miner is when an honest miner generates a block. If the acceptance window $w$ is at least $\Delta_B$, any honest miner will receive a block that is generated by another honest miner before their acceptance window ends. This is because $\Delta_B$ is greater than the 100th percentile propagation time of the block.
\end{proof}
Considering the aforementioned theorem, the acceptance window $w$ is set to $\Delta_B$ in the proposed method. 

Next, we explain the parameter settings for a sufficiently shared partial PoW.
\begin{thm}
     If the time considered for a valid partial PoW to be sufficiently shared is at least $w + \Delta_B$, a partial PoW that is published by an attacker after an honest miner generates a block will not be considered as sufficiently shared by any honest miner before their acceptance window ends.
\end{thm}
\begin{proof}
    Regardless of how the block that is generated by an attacker is propagated, the acceptance window for any honest miner starts within $\Delta_B$ after an honest miner generates a block. Therefore, the acceptance windows of all miners will end within $w + \Delta_B$ of an honest miner generating a block. Consequently, if the time that is considered for a valid partial PoW to be sufficiently shared is at least $w + \Delta_B$, no honest miner will regard a partial PoW that is published after an honest miner generates a block as sufficiently shared within their acceptance window.
\end{proof}
Considering the aforementioned theorem and the acceptance window $w = \Delta_B$, a valid partial PoW is considered as sufficiently shared $2\Delta_B$ after it arrives in the proposed method. 

Next, we examine the parameter settings for shared partial PoW.
\begin{thm}
     A valid partial PoW for which $w + \Delta_B + \Delta_P$ has elapsed since its arrival will be sufficiently shared for all honest miners.
\end{thm}
\begin{proof}
    When an honest miner receives a partial PoW, all honest miners receive that partial PoW within $\Delta_P$. Therefore, after more than $w + \Delta_B + \Delta_P$ has elapsed since the honest miner received the partial PoW, the partial PoW is considered to be sufficiently shared by all honest miners.
\end{proof}
Considering the aforementioned theorems and the acceptance window $w = \Delta_P$, our proposed method includes a valid partial PoW in the shared partial PoW $2\Delta_B + \Delta_P$ after its arrival.

Note that if even one of the shared partial PoWs in the blocks that constitute a chain is not sufficiently shared, a value of $-1$ is assigned to that chain. This ensures that honest miners can always distinguish between the attacker's block and the block generated by an honest miner when an attacker who has generated a block but has not published its shared partial PoW triggers Match against post-generated block.

\section{Evaluation} \label{evaluation}
In this section, we evaluate the effectiveness of our proposed method. In Section \ref{matchagainstpost-generatedblock}, we demonstrate the effectiveness of our method in suppressing Match against post-generated block. This section includes both theoretical and simulation analyses.

In Section \ref{matchagainstpre-generatedblock}, we evaluate the proposed method by using ESM, which is a straightforward extension of SM that also considers Match against pre-generated block. We investigated the difference in the relative revenue for attackers between ESM and SM by varying the acceptance window. The analysis results demonstrate the importance of setting the acceptance window for a last-generated rule.

\begin{figure*}[tb]
    \centering
    \begin{subfigure}[b]{0.3\textwidth}
        \centering
        \includegraphics[width=\textwidth]{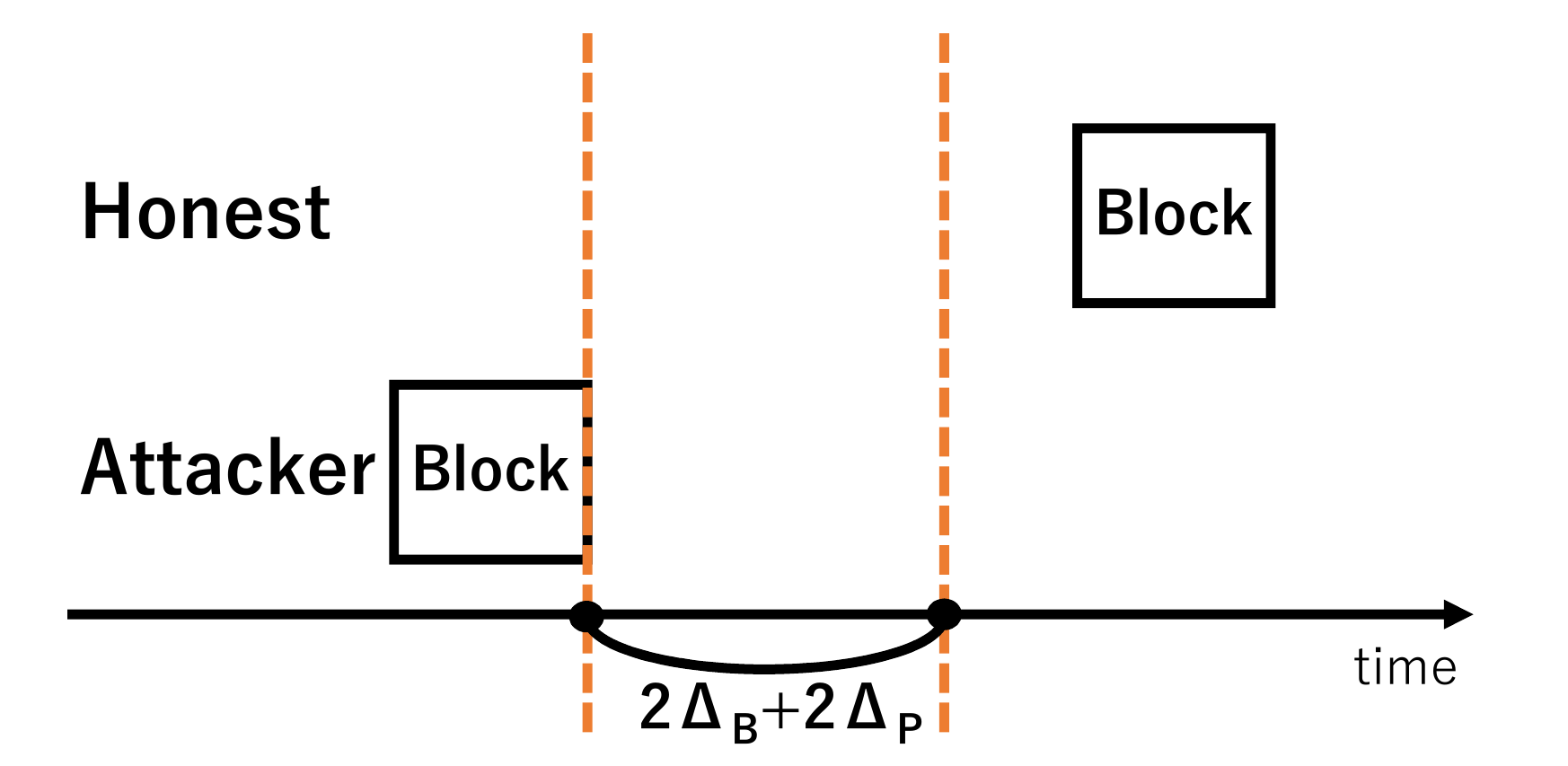}
        \caption{Situation that honest miners even.}
        \label{PPoWEVEN}
    \end{subfigure}
    \hfill 
    \begin{subfigure}[b]{0.3\textwidth}
        \centering
        \includegraphics[width=\textwidth]{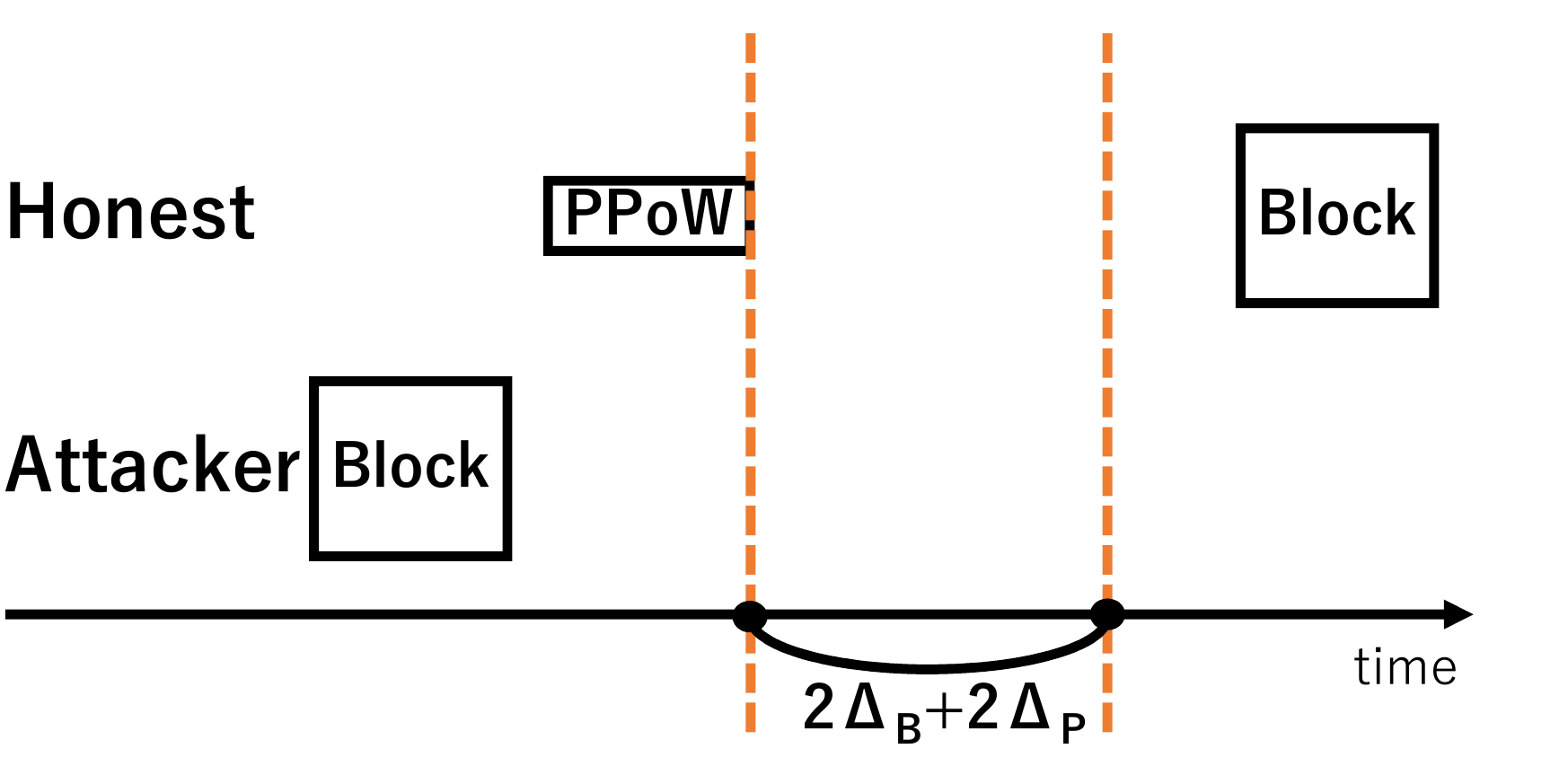}
        \caption{Situation that honest miners win.}
        \label{PPoWWIN}
    \end{subfigure}
    \hfill
    \begin{subfigure}[b]{0.3\textwidth}
        \centering
        \includegraphics[width=\textwidth]{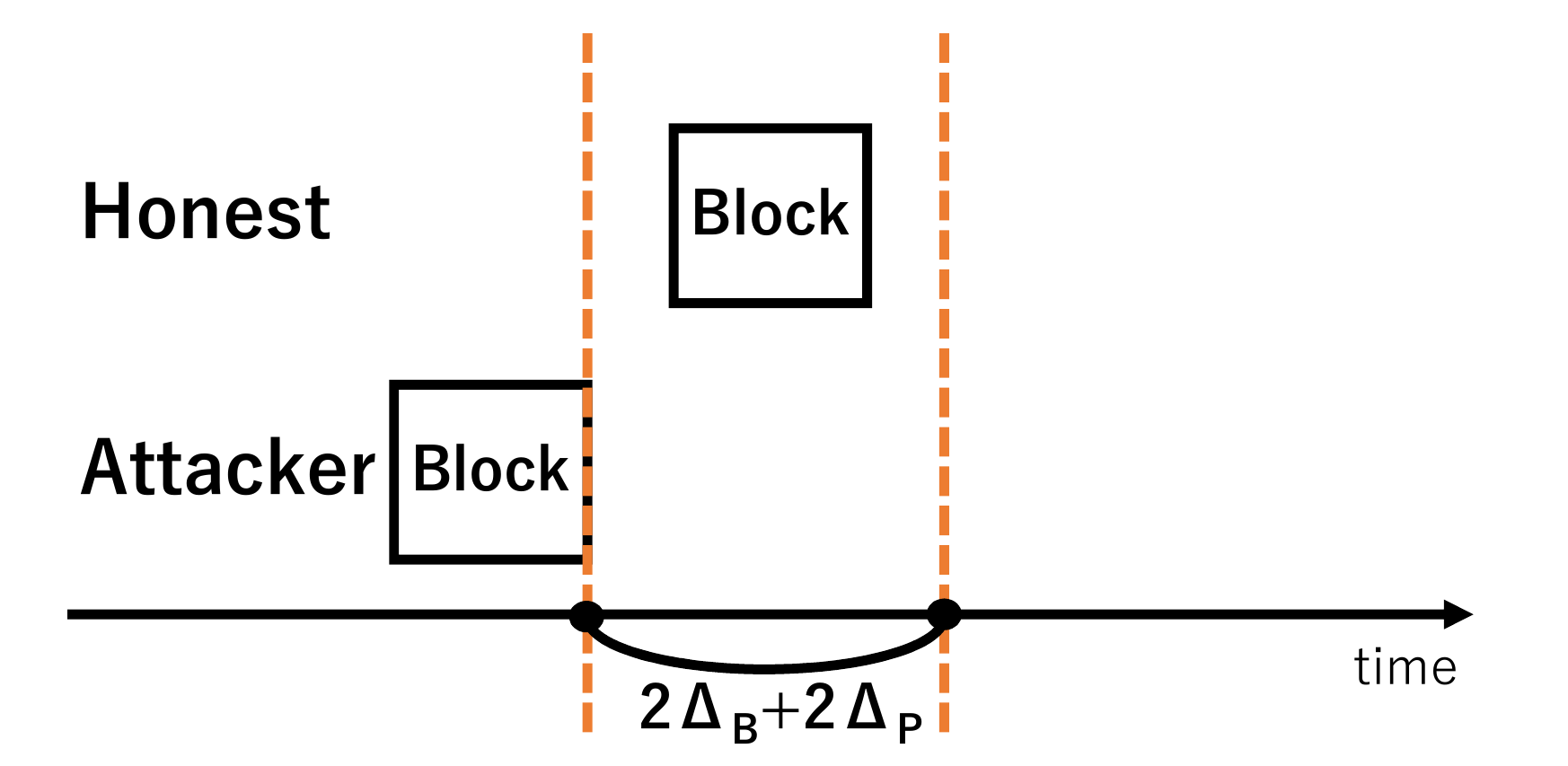}
        \caption{Situation that honest miners may lose.}
        \label{PPOWLOSE}
    \end{subfigure}

    \caption{If the attacker successfully generates a block, they must publish their partial PoW. This is because if an insufficiently shared PoW is included in a block, honest miners will assign a value of $-1$ to the block, leading to the attacker's block not being chosen in the event of a tie. It takes $2\Delta_B + 2\Delta_P$ for honest miners to include one partial PoW after it has been published. This duration includes $\Delta_P$ for the partial PoW to be propagated across the network and $2\Delta_B + \Delta_P$ for each miner to consider it sufficiently shared and include it in a block. Figure \ref{PPoWEVEN} shows the situation in which an honest miner's block is chosen randomly in a tie. This is because the attacker's shared partial PoWs that have been published are included in the block generated by the honest miner. Figure \ref{PPoWWIN} shows the situation in which an honest miner's block is chosen in a tie. Figure \ref{PPOWLOSE} shows the situation in which an honest miner's block may not be chosen in a tie. This is because the attacker's partial PoW that has been published may not be included in blocks honest miners are mining.}
    \label{PPoWSituation}
\end{figure*}

\subsection{Suppression of Match against post-generated block}
\label{matchagainstpost-generatedblock}
\subsubsection{Theoretical Analysis}
We theoretically demonstrate the effectiveness of our proposed method in suppressing Match against post-generated block. Specifically, we evaluate the proportion $\gamma$ of honest miners who mine on the attacker's chain when the attacker performs Match against post-generated block. $\gamma$ is the most important general indicator. Most attacks that use intentional forks, such as Selfish Mining, also use intentional chain ties. Suppressing intentional chain ties is the same as suppressing $\gamma$.

In this case, we assume that all honest miners adopt the proposed method. 

Let $\alpha$ denote the proportion of the attacker's hashrate and $T$ denote the average block generation interval across the entire network.
First, the following lemma holds:
\begin{lem} \label{lemma1}
    $\gamma$ satisfies the following inequality, on average:
    \begin{align*}
    \gamma \leq 1 -  \frac{n - 1}{n + \frac{\alpha}{1 - \alpha}} \exp (-\frac{2 \Delta_B + 2 \Delta_P}{T}).
    \end{align*}
\end{lem}

\begin{proof}
    As explained in Section \ref{parameter}, in a system in which all honest miners adopt the proposed method, if an attacker initiates a Match against post-generated block, they must publish the shared partial PoW corresponding to their block at the time of block generation. Otherwise, honest miners will not select the attacker's block. 
    
    In a system in which all honest miners adopt our proposed method and an attacker initiates Match against post-generated block, the following conditions (Condition (1) and Condition (2)) can be considered sufficient for the chain that is generated by honest miners to be selected among the chains in a tie (Figure \ref{PPoWWIN}).
    \begin{description}
    \item[\textbf{Condition (1)}:] \mbox{}\\ An honest miner generates a partial PoW after the attacker has generated a block and before the next block is generated.
    \item[\textbf{Condition (2)}:] \mbox{}\\ Before the next block is generated, all honest miners incorporate the partial PoW that is generated under Condition (1) into the shared partial PoW of their block and proceed with mining.
    \end{description}

    The attacker does not need to publish a partial PoW after generating a block for Match against post-generated block. Thus, the probability of Condition (1), which is denoted as $Pr[ Cond(1)]$, satisfies the following inequality:
    \begin{align}
        Pr[ Cond (1)] &\geq \Sigma_{k = 0}^\infty \frac{(1 - \alpha) (n - 1)}{n}(\frac{\alpha (n - 1)}{n})^k \\
        &= \frac{n - 1}{n + \frac{\alpha}{1 - \alpha}}
    \end{align}
    The probability of the attacker first generating a partial PoW k times, followed by an honest miner generating a partial PoW, is expressed as $\frac{(1 - \alpha) (n - 1)}{n}(\frac{\alpha (n - 1)}{n})^k$. The right side of the inequality represents the probability that the attacker does not publish partial PoWs. Therefore, the value is higher when an attacker publishes partial PoWs.

    The probability of Condition (2), which is denoted as $Pr[ Cond(2)]$, can be expressed as follows:
    \begin{align}
      Pr[ Cond(2)] &= (1 - p)^{N_i}\label{genetime}\\
              &\approx \exp (-p N_i) \label{approx1}\\
              &= \exp (-p (\Sigma_{j \in V} P_{ij} M_j + (2 \Delta_B + \Delta_P)M_{all})) \label{waitime} \\
              &=\exp (-\Sigma_{j \in V} \frac{P_{ij}}{T} \frac{M_j}{M_{all}}  - \frac{2 \Delta_B + \Delta_P}{T}) \label{pin} \\
              &\geq \exp (-\frac{2 \Delta_B + 2 \Delta_P}{T}) \label{round}
    \end{align}
    where $V$ denotes the set of miners participating in the system. $N_i$ represents the sum of the hash computations by miners who do not include a partial PoW generated by miner $i$ into shared partial PoWs. In addition, $p$ is the probability of generating a valid block in a single hash calculation. $P_{ij}$ indicates the time that is required for miner $j$ to receive the partial PoW after it has been successfully generated by miner $i$. $M_i$ represents the hashrate of miner $i$.

    Equation \ref{genetime} represents the probability that a valid block is not generated after $N_i$ hash calculations. Approximate equation \ref{approx1} is an approximation of Equation \ref{genetime}, based on $p << 1$ and $N_i >> 1$. Equation \ref{waitime} is derived by substituting $N_i = (\Sigma_{j \in M} P_{ij} M_j + (2 \Delta_B + \Delta_P)M_{all}$ into Equation \ref{approx1}. Here, $P_{ij} M_j$ is the total number of hash calculations performed by miner $j$ until the partial PoW arrives. The second term on the right side indicates that it takes each miner $2 \Delta_B + \Delta_P$ to include the received partial PoW in their shared partial PoW and start mining. Equation \ref{pin} is obtained by substituting $p = 1 / {M_{all} T}$. Inequality \ref{round} holds under the condition $P_{ij} < \Delta_P$.

    From the aforementioned, the probability (Pr[Cond(1)]Pr[Cond(2)]) that an honest miner selects a block that is generated by an honest miner instead of one that is generated by the attacker satisfies the following inequality:
    \begin{align}
     Pr[Cond(1)]Pr[Cond(2)] \geq \frac{n - 1}{n + \frac{\alpha}{1 - \alpha}} \exp (-\frac{2 \Delta_B + 2 \Delta_P}{T})
    \end{align}

    Therefore, when the attacker initiates Match against post-generated block, the probability that honest miners will select the attacker's chain is less than or equal to$\frac{n - 1}{n + \frac{\alpha}{1 - \alpha}} \exp (-\frac{2 \Delta_B + 2 \Delta_P}{T})$. Consequently, the average value of $\gamma$ satisfies the following inequality:
    \begin{align}
    \gamma &\leq 1 -  Pr[Cond(1)]Pr[Cond(2)] \\
          &\leq 1 -  \frac{n - 1}{n + \frac{\alpha}{1 - \alpha}} \exp (-\frac{2 \Delta_B + 2 \Delta_P}{T}) 
    \end{align}
\end{proof}

The term $\Sigma_{j \in V}\frac{P_{ij}M_j}{M_{all}}$ in Equation \ref{pin} represents the hashrate-weighted average propagation time of a partial PoW when it is generated by miner $i$. It is widely known that the block propagation time and fork rate are deeply connected\cite{informationpropagation}. The relationship between the hashrate-weighted propagation time $T_W$ and the fork rate $F$ can be expressed as $F = T_W / T$\cite{ImpactoftheHashRateontheTheoreticalForkRateofBlockchain}. The fork rate refers to the probability of another miner creating a block that forks the blockchain while a certain block is still propagating through the network. Therefore, knowing the fork rate enables the estimation of the hashrate-weighted propagation time of a block. Although determining the current value of Bitcoin is challenging, previous research\cite{onTheSecurity} has shown the fork rate of Bitcoin to be approximately 0.41\%. This implies a hashrate-weighted propagation time of approximately 2.4 s, which is significantly less than the measured 100th percentile propagation time of a block\cite{statsofbitcoin}. When applying our proposed method to Bitcoin, each partial PoW would be 80 bytes, which is considerably smaller than a block, suggesting that the hashrate-weighted propagation time of partial PoW could be even smaller. Consequently, the actual value of $Pr[Cond(2)]$ is likely to be closer to $\exp (\frac{2 \Delta_B + \Delta_P}{T})$ than the value in inequality \ref{round}.

Next, we examine the scenario in which a tie occurs. Specifically, the following lemma holds:
\begin{lem}\label{lemma2}
    $\gamma$ satisfies the following inequality, on average:
    \begin{align*}
    \gamma \leq 1 -  \frac{1}{2} \exp (-\frac{2 \Delta_B + 2 \Delta_P}{T}).
    \end{align*}
\end{lem} 
\begin{proof}
    In a system in which all miners adopt the proposed method, upon generating a block an attacker must publish shared partial PoWs that belong to the set corresponding to the shared partial PoW of the block. $2\Delta_B + 2 \Delta_P$ after the attacker publishes these partial PoWs, all honest miners include the attacker's published partial PoWs in their shared partial PoW and engage in mining. Thus, if no new blocks are generated within $2\Delta_B + 2 \Delta_P$ after the attacker has generated a block, the honest miners will have at least one shared partial PoW of the same size as the attacker's chain (Figure \ref{PPoWEVEN}). In addition, a chain is randomly selected in the proposed method if the number of shared partial PoWs is equal. Considering these factors, the following inequality is established:
    \begin{align}
     \gamma \leq 1 -  \frac{1}{2} \exp (-\frac{2 \Delta_B + 2 \Delta_P}{T}) 
\end{align}
\end{proof}
According to Lemmas \ref{lemma1} and \ref{lemma2}, the following theorem holds:
\begin{thm}\label{third1}
$\gamma$ satisfies the following inequality, on average:
\begin{align*}
    \gamma &\leq \min( 1 -  \frac{1}{2} \exp (-\frac{2 \Delta_B + 2 \Delta_P}{T}) , \\& 1 -  \frac{n - 1}{n + \frac{\alpha}{1 - \alpha}} \exp (-\frac{2 \Delta_B + 2 \Delta_P}{T})).
\end{align*}
\end{thm}

\begin{figure}[tb]
\begin{center}
\includegraphics[width=0.5\textwidth]{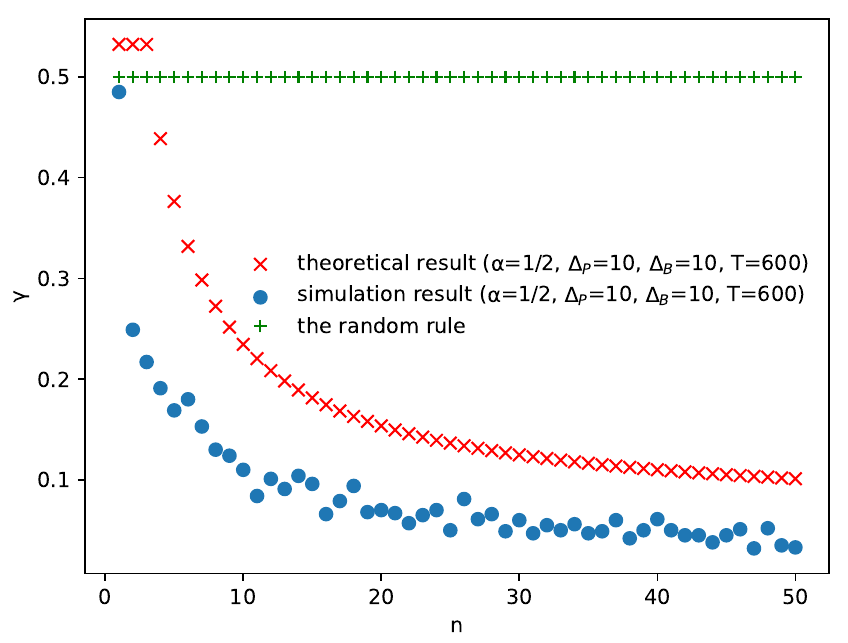}
\end{center}
\caption{Impact of difficulty adjuster $n$ on effectiveness of proposed method with $\alpha = 0.5$, $T = 600$, and $\Delta_B = \Delta_P = 10$.}
\label{gamma1}
\end{figure}

\begin{figure}[tb]
\begin{center}
\includegraphics[width=0.5\textwidth]{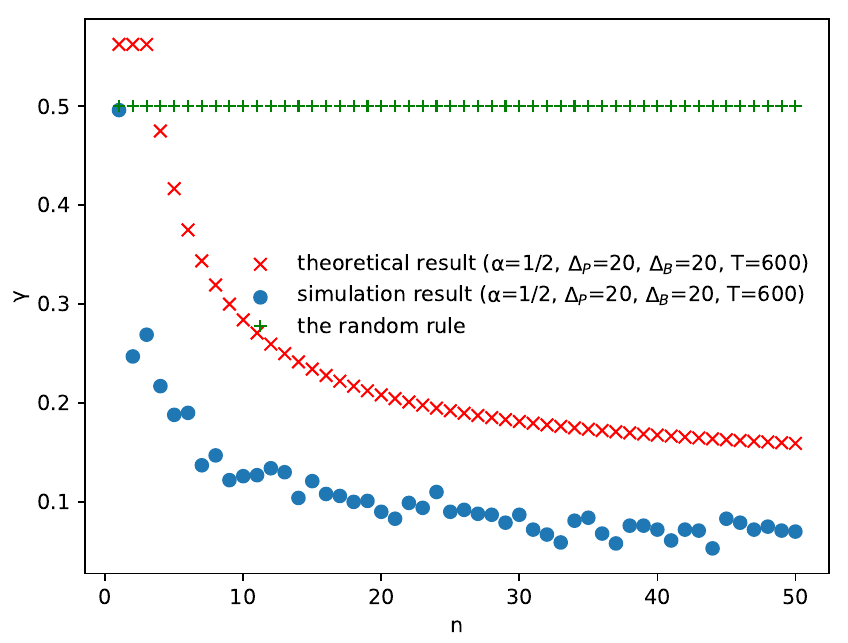}
\end{center}
\caption{Impact of difficulty adjuster $n$ on effectiveness of proposed method with $\alpha = 0.5$, $T = 600$, and $\Delta_B = \Delta_P = 20$.}
\label{gamma2}
\end{figure}

Next, we quantitatively analyze the proposed method.
We consider two parameter setting scenarios: $\Delta_B = \Delta_P = 10$ and $\Delta_B = \Delta_P = 20$. The first scenario, where $\Delta_B = \Delta_P = 10$, is based on the current block propagation time in Bitcoin\cite{statsofbitcoin}. The second scenario, with $\Delta_B = \Delta_P = 20$, is selected to assess the performance of the proposed method under more relaxed settings and to account for the impact of clock drift in the system. This approach enables a comprehensive evaluation of the effectiveness of the method under various network conditions and parameter configurations.

Figures \ref{gamma1} and \ref{gamma2} depict the variations in $\gamma$ for different difficulty adjusters $n$, given the proportion of an attacker's hashrate $\alpha$ of 0.5 and the average block generation interval $T$ of 600. Figure \ref{gamma1} shows the case in which $\Delta_B = \Delta_P = 10$, whereas Figure \ref{gamma2} shows the results for $\Delta_B = \Delta_P = 20$. 
As is evident from the figures, the value of $\gamma$ decreases as the difficulty adjuster $n$ increases. However, increasing $n$ also implies an increase in communication overhead. Thus, varying the difficulty adjuster according to the condition of the system is crucial. Additionally, the figures reveal that our proposed method reduces $\gamma$ more effectively than the random rule when $n$ is greater than 3. 
Specifically, for $n$ = 50 with $\Delta_B = \Delta_P = 20$, $\gamma$ is reduced to 0.10118 or less, and for $\Delta_B = \Delta_P = 20$, $\gamma$ is reduced to 0.15915 or less.

\begin{figure}[tb]
\begin{center}
\includegraphics[width=0.5\textwidth]{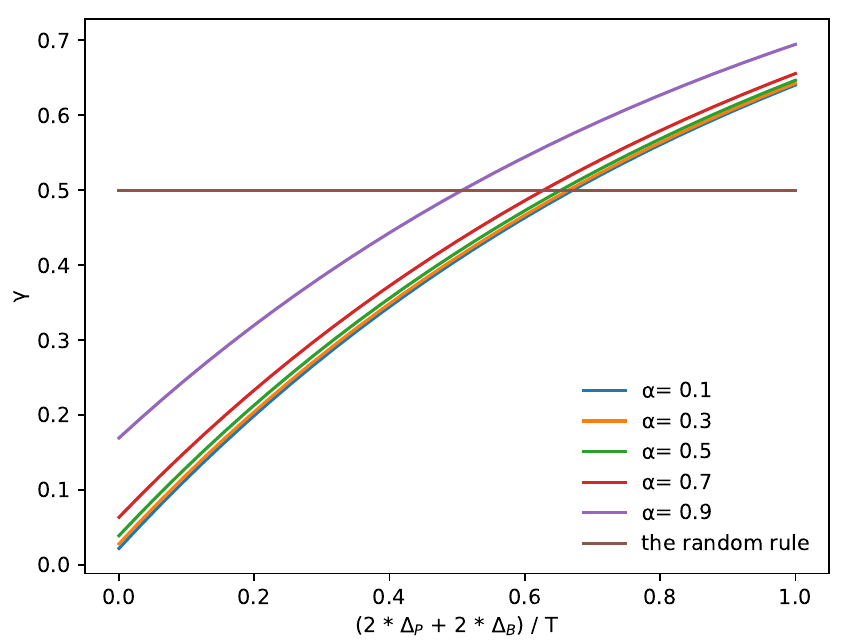}
\end{center}
\caption{Impact of $\frac{2\Delta_B + 2\Delta_P}{T}$ on effectiveness of proposed method at $n = 50$.}
\label{alpha}
\end{figure}

Figure \ref{alpha} depicts the effectiveness of the proposed method when the difficulty adjuster $n$ is fixed at 50 and the value $\frac{2\Delta_B + 2\Delta_P}{T}$ is varied. As illustrated in the figure, the effectiveness of the proposed method decreases as $\frac{2\Delta_B + 2\Delta_P}{T}$ increases. However, it is also apparent that the proposed method is more effective than the random approach for systems with shorter block generation intervals than Bitcoin, such as Litecoin (with a block generation interval of 2.5 min)\cite{Litecoin2023}.
The figure also indicates that the effectiveness of our method decreases as the proportion of the attacker's hashrate increases. This diminished effectiveness is attributed to the reduction in the number of partial PoWs that are publicly available on the network. Increasing the difficulty adjuster $n$ can easily address this issue.

\begin{figure}[tb]
\begin{center}
\includegraphics[width=0.5\textwidth]{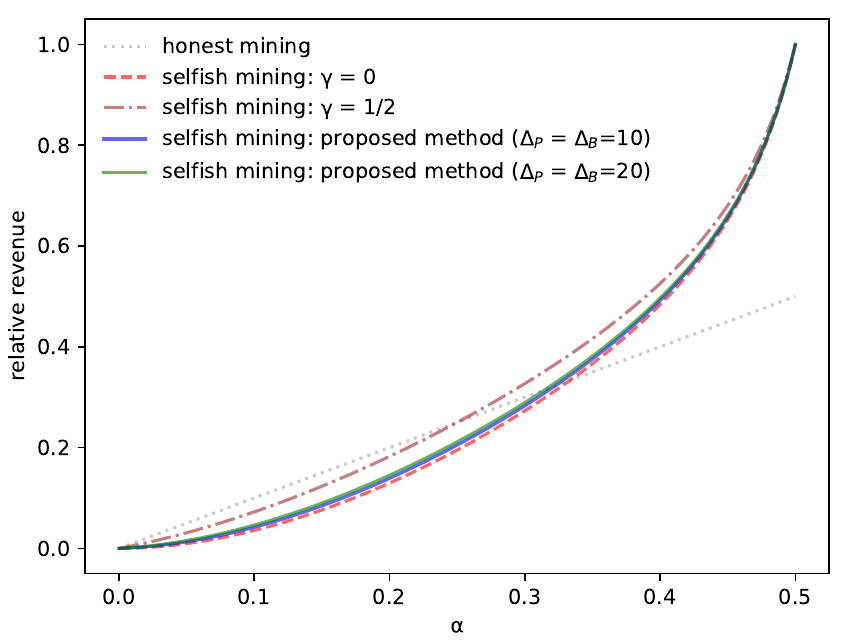}
\end{center}
\caption{Analysis of relative revenue of the attacker in SM under proposed method at $n = 50$, $T = 600$.}
\label{threshold}
\end{figure}
SM is the most important and simple example that uses intentional chain ties. Figure \ref{threshold} demonstrates the relative revenue of the attacker performing SM in a system in which all miners adopt our proposed method, with the difficulty adjuster $n$ fixed at 50 and the block generation interval $T$ set to 600. The blue line represents scenarios with $\Delta_B = \Delta_P = 10$, and the green line represents cases with $\Delta_B = \Delta_P = 20$. For $\Delta_B = \Delta_P = 10$, the threshold of the proportion of the attacker's hashrate at which the relative revenue of the attacker executing SM becomes larger than that in the normal protocol is approximately $\alpha = 0.32246$. For $\Delta_B = \Delta_P = 20$, the threshold is approximately $\alpha = 0.31479$. The random rule is another technique that can be easily integrated into existing systems. However, when the random rule is applied, the threshold of the proportion of the attacker's hashrate at which their revenue for executing SM surpasses that in the normal protocol is established at $\alpha = 0.25$. This indicates that our proposed method is more effective than the random rule in suppressing SM from the perspective of reducing the relative revenue of the attacker.

\subsubsection{Simulation Analysis}
We further evaluate our method through simulation analysis.
The evaluation of the effectiveness of the method in Lemma \ref{lemma1} only considers situations in which all honest miners hold a shared partial PoW that is larger than the shared partial PoW of the attacker's block. However, as explained after the proof of Lemma \ref{lemma1}, a more stringent evaluation is possible by considering the hashrate-weighted propagation time of partial PoW, which is not considered in Lemma \ref{lemma1}. In addition, the current assessment fails to consider the effect of each miner holding a shared partial PoW of the same size as the attacker's block, as indicated in the evaluation presented in Lemma \ref{lemma2}.
Another overlooked aspect is the scenario in which honest miners successfully generate a block while some miners have not yet received the partial PoW that is published by the attacker. In such cases, miners, who have not received the attacker's partial PoW, would select the block that is generated by an honest miner instead of the attacker's block during a chain tie. This factor is also not considered in the evaluation presented in Lemma \ref{lemma2}.

A theoretical evaluation that accounts for all factors is challenging. In addition, conducting experiments in a real network is also challenging. Therefore, we conducted a simulation experiment to assess the proposed method. The simulation allows us to consider both scenarios: when all honest miners possess a shared partial PoW that is larger than the attacker's shared partial PoW, and when they have a shared partial PoW that is the same size as that of the attacker. We used SimBlock\cite{simblock}, which is a tool that can simulate a blockchain network. The default network parameters of SimBlock for the year 2019 were used. For simplicity, we did not apply Compact Block Relay\cite{bip152}, which is the default setting in SimBlock. The number of miners was set to 300, and the block generation interval $T$ was set to 600 s. The block size was adjusted to approximately 200,000 bytes and 500,000 bytes so that the 100\%ile propagation time of the block was approximately 10 s and 20 s. For simplicity, the size of partial PoW was also set to 200,000 bytes and 500,000 bytes.

We randomly selected one miner from the network participants to act as the attacker in the simulation. This attacking miner was set to possess half of the total hashrate of the entire network. The attacker could receive messages from honest miners instantly and when the attacker published a message, it was instantly and simultaneously delivered to all other miners.

In our simulation, the attacker basically did not publish a partial PoW to the network, even after generating it. Upon receiving a partial PoW, the attacker instantly incorporated it into their shared partial PoW and continued mining. In addition, when the attacker successfully generated a block, they did not publish it immediately, but continued mining. The attacker published the withheld partial PoW to the network after a delay of $\Delta_B$ from successful block generation. If the attacker consecutively succeeded in generating blocks and the difference between the chain of the attacker and honest miners reached two, the attacker published both their block and the shared partial PoW. When an honest miner succeeded in generating a block, the attacker instantly initiated a Match against post-generate block.

For simplicity, in our simulation, we assumed that honest miners did not check whether the set of partial PoWs corresponding to the attacker's shared partial PoW was sufficiently shared during a chain tie. However, unlike the attacker, honest miners included a partial PoW in their shared partial PoW for mining only after $2\Delta_B + \Delta_P$ had elapsed since receiving the partial PoW.

The effectiveness of the proposed method in the real network is better than the simulation result from the viewpoint of suppressing intentional chain ties. This is because the settings for the simulation described above were configured to be advantageous for the attacker. For example, the actual size of a partial PoW is 80 bytes, but we increased it to 200,000 and 500,000 bytes in our simulation. This enlargement inevitably lengthened the propagation time of the partial PoW, thereby reducing the likelihood that honest miners would hold a partial PoW larger than that of the attacker, compared to actual circumstances. Moreover, attackers in our simulation could instantly incorporate any received partial PoW into their shared partial PoW. This action could potentially lead to the inclusion of insufficiently shared partial PoWs in the shared partial PoW. However, in our setup, honest miners did not verify whether the partial PoWs were sufficiently shared. This increased the probability that the attacker’s shared partial PoW would be larger than that of the honest miners. In addition, although the attacker published their shared partial PoW upon block generation, this publication only occurred after a delay of $\Delta_P$. In reality, honest miners may generate a block during the propagation of shared partial PoW, and some honest miners may be unaware of this shared partial PoW. However, this risk was disregarded in our simulation.

Figures \ref{gamma1} and \ref{gamma2} show the results of the simulation experiments of 1000 times of chain ties. The data clearly show that the effectiveness of the proposed method, as measured by these simulations, exceeded the values obtained in the theoretical analysis. Specifically, when n = 50 and $\Delta_B = \Delta_P = 10$, the value of $\gamma$ is 0.033. In contrast, for $\Delta_B = \Delta_P = 20$, $\gamma$ increased to 0.07.

\subsection{Suppression of Match against pre-generated block}
\label{matchagainstpre-generatedblock}
\begin{figure}[tb]
\begin{center}
\includegraphics[width=0.5\textwidth]{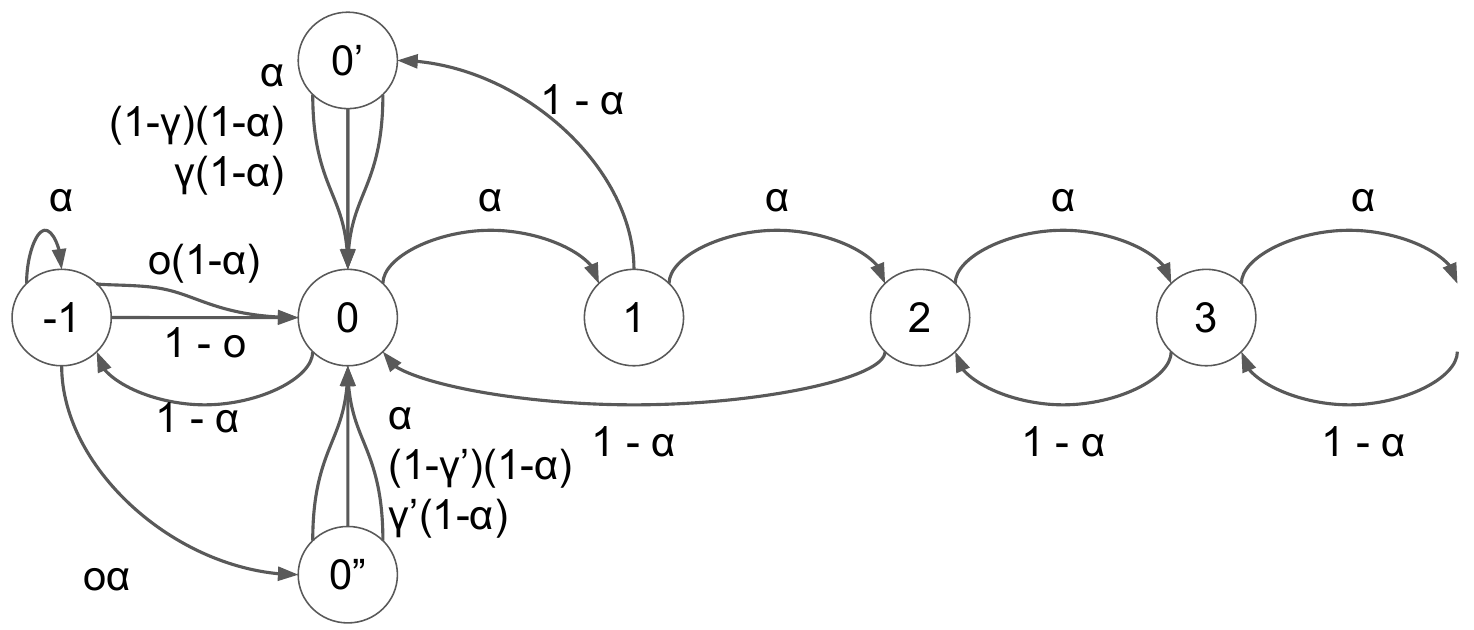}
\end{center}
\caption{State transition diagram for extended selfish mining.}
\label{extended selfish mining}
\end{figure}

As noted in Section \ref{intentionalfork}, introducing an acceptance window is necessary to suppress Match against pre-generated block attacks initiated by attackers in last-generated rules, including our proposed method. This section quantitatively demonstrates the effect of the acceptance window parameter settings on the effectiveness of the proposed method. Specifically, we introduce ESM that considers Match against pre-generated block.

ESM is an extension of SM that incorporates Match against pre-generated block.
When honest miners and an attacker that is engaged in ESM share the same blockchain, the attacker continues to mine on the same block even after an honest miner has successfully generated a block. The duration of this continued mining by the attacker is known as the unresponsive time $s$. If the attacker can generate a block within the unresponsive time, a Match against pre-generated block occurs. If the attacker fails to generate a block within the unresponsive time, they accept and mine on the honest miner's chain. The same acceptance and subsequent mining on the honest miner's chain occurs if the honest miner generates a block within the unresponsive time. The remaining mining strategy for the attacker is the same as that in SM.

We modeled a network in which an attacker performed ESM using a Markov decision process. We assumed that honest miners adopted the proposed method. For simplicity, we assumed that natural chain ties did not occur. The reward for each miner was based on the number of blocks that they generated in the main chain. We set the average block-generation interval $T$ to 600 for these evaluations.
Each state was determined by the difference in the number of blocks that were mined by the attacker and honest miners, and whether a chain tie had occurred. State transitions occurred based on the block generation by each miner or the time elapsed since the block generation.

The specific states were as follows:
\begin{align}
-1, 0, 0', 0'', 1, 2, 3, \ldots
\end{align}
The number of each state indicates the difference in the block height between the blocks that were mined by the attacker and honest miners. Figure \ref{extended selfish mining} shows the state transition diagram for ESM. The base state is $0$, in which both honest miners and attackers mine the same chain. When an honest miner generates a block in state $0$, the system transitions to state $-1$. In state $-1$, the attacker continues to mine the same block even after an honest miner has generated a block. If the attacker successfully generates a block within the unresponsive time in state $-1$, a Match against pre-generated block is triggered, leading to state $0''$. In state $0''$, a proportion $\gamma'$ of honest miners selects and mines on the attacker's chain. The system transitions from state $0''$ to state $0$ after a new block is generated. If the unresponsive time is exceeded in state $-1$, the attacker assumes that winning a chain tie is impossible, and accepts the honest miner's chain, thereby transitioning back to state $0$. The probability that a block is generated within the unresponsive time is denoted by $o$. States $1, 2, \ldots$ represent scenarios in which the chain of the attacker is longer than that of honest miners. In this case, the attacker does not publish their chain. If an honest miner generates a block in state $1$, the attacker publishes their block and triggers a Match against post-generated block. This transition moves from state $1$ to state $0'$, where the proportion $\gamma$ of honest miners selects and mines on the chain of the attacker. The system moves from state $0'$ to state $0$ after a new block is generated. If an honest miner generates a block in state $2$, the attacker publishes all their blocks, thereby invalidating the honest miner's chain.

The value of $o$ is defined as
\begin{align}
o = 1 - \exp(- \frac{s}{T})
\end{align}
where $o$ represents the probability that a new block is generated by the network within the duration of $s$, which is the time from when the attacker receives the block from an honest miner to the end of the unresponsive time. 

The relative revenue $R$ for the attacker performing ESM is calculated as follows:
\tiny
\begin{align}
R =& \frac{o\alpha(1-\alpha)(1 - 2 \alpha)(2\alpha + \gamma'(1-\alpha))}{(1 - o(1-\alpha))(\alpha^3 - 4\alpha^2 + 2 \alpha) + (1+o\alpha)(1-2\alpha)(1-\alpha)} \nonumber\\
& + \frac{(1 - o(1-\alpha))(4\alpha^4 + -9\alpha^3 + 4\alpha^2 + \gamma \alpha(1-2\alpha)(1-\alpha)^2))}{(1 - o(1-\alpha))(\alpha^3 - 4\alpha^2 + 2 \alpha) + (1+o\alpha)(1-2\alpha)(1-\alpha)}
\end{align}
\normalsize

Candidates for the unresponsive time $s$ of the attacker include the end of the acceptance window of the block generator $\Delta_B$ and the end of the acceptance window for all miners $2\Delta_B$. If the unresponsive time is defined as the duration until the acceptance window of the block generator ends, $\gamma'$ will always be 1.
However, if the unresponsive time is defined as the duration until the acceptance window for all miners ends, $\gamma'$ will not necessarily be 1 even if the attacker successfully generates a block and causes a chain tie. In the following analysis, we assume that $\gamma' = 1$, which is advantageous for the attacker and yields tighter results.

\begin{figure}[tb]
\begin{center}
\includegraphics[width=0.5\textwidth]{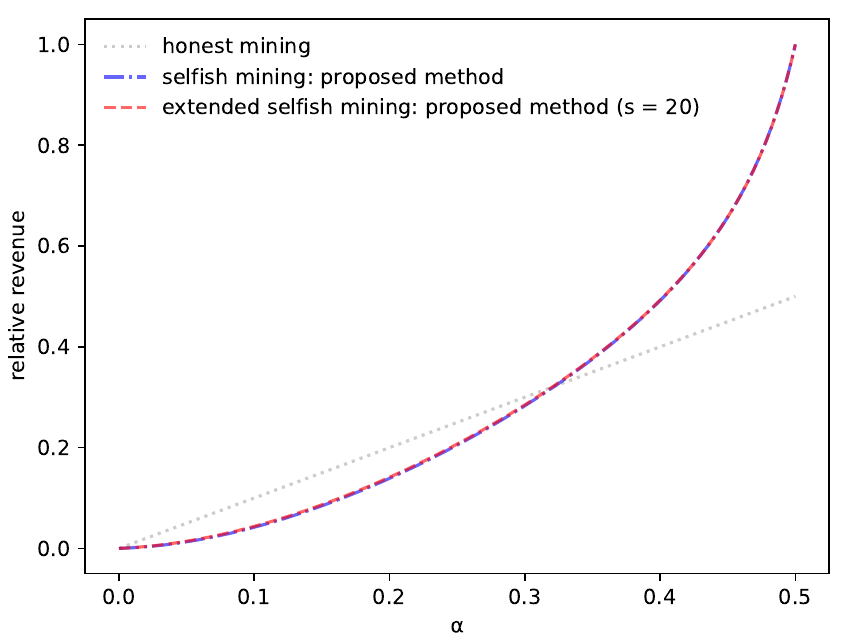}
\end{center}
\caption{Analysis of relative revenue of the attacker in ESM under proposed method at $s = 20$, $n = 50$, $T = 600$ and $\Delta_B = \Delta_P = 10$.}
\label{$s = 20$}
\end{figure}

\begin{figure}[tb]
\begin{center}
\includegraphics[width=0.5\textwidth]{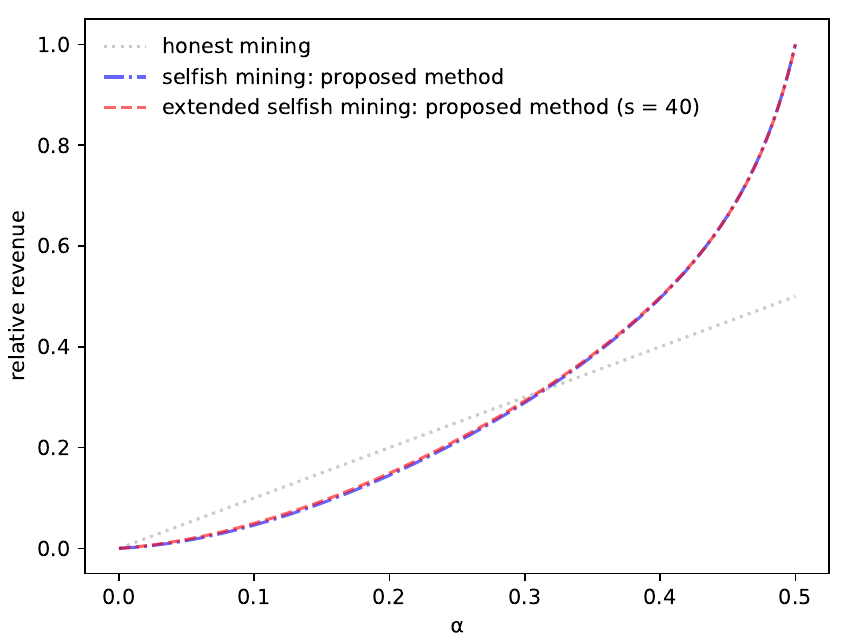}
\end{center}
\caption{Analysis of relative revenue of the attacker in ESM under proposed method at $s = 40$, $n = 50$, $T = 600$ and $\Delta_B = \Delta_P = 20$.}
\label{$s = 40$}
\end{figure}

We analyzed with various settings for the unresponsive time and different values of $\Delta_B$ and $\Delta_P$. The results are shown in Figures \ref{$s = 20$} and \ref{$s = 40$}. As seen in the figures, there was almost no difference in the relative revenue of the attacker between the SM and ESM. The required proportion of the attacker's hashrate for successful SM was approximately $\alpha = 0.32246$ with $s = 20$ and $\Delta_B = \Delta_P = 10$, and approximately $\alpha = 0.31479$ with $s = 40$ and $\Delta_B = \Delta_P = 20$, whereas it was approximately $\alpha = 0.32045$ and $\alpha = 0.31055$, respectively, ESM. These results demonstrate that setting the acceptance window $w$ appropriately can mitigate Match against pre-generated block.

\begin{figure}[tb]
\begin{center}
\includegraphics[width=0.5\textwidth]{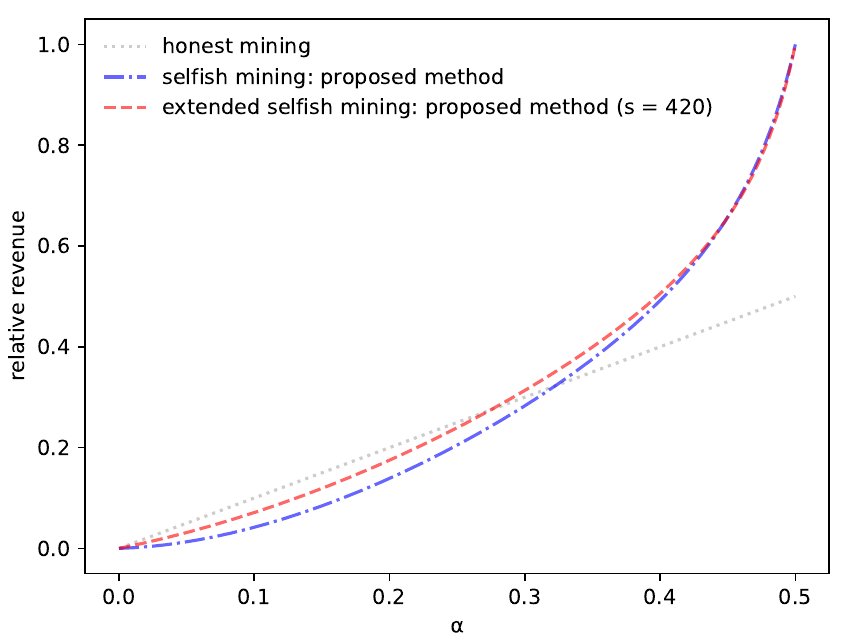}
\end{center}
\caption{Analysis of relative revenue of the attacker in ESM under proposed method at $s = 420$, $n = 50$, $T = 600$ and $\Delta_B = \Delta_P = 10$.}
\label{$s = 420$}
\end{figure}

We also analyzed cases with extremely long acceptance windows. Figure \ref{$s = 420$} shows the relative revenue of the attacker when the unresponsive time was set to $s = 420$. It can be observed that the attacker achieved a higher relative revenue with ESM than with standard SM. The necessary proportion of the attacker's hash rate for a successful attack in this scenario was approximately 0.27469. This result indicates that failing to set the acceptance window appropriately increases the risk of Match against pre-generated block, thereby compromising the network security.
\section{Implementation Details} \label{detail}
\subsection{Application to Bitcoin}\label{applicationToSystems}
In this section, we discuss the application of our proposed method to existing systems, particularly Bitcoin, and the associated effects. The method that we discuss here is the n-variable method.

\begin{figure}[tb]
\begin{center}
\includegraphics[width=0.5\textwidth]{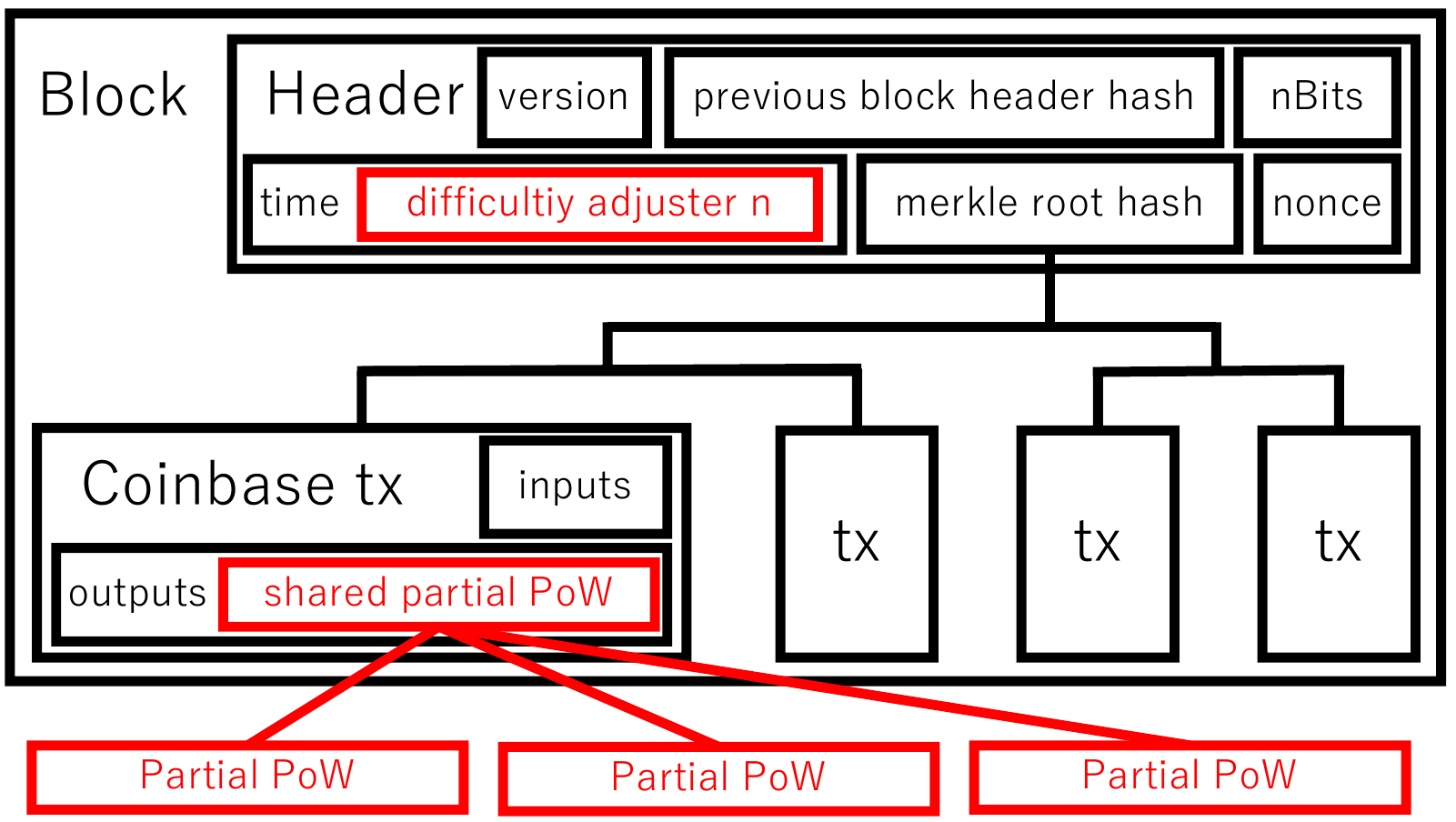}
\end{center}
\caption{Data structure of a block with the proposed method. The data enclosed in a red frame represent new data required for the specifications of the proposed method. Each partial PoW is the block header whose hash value meets the difficulty determined by the difficulty adjuster.}
\label{structure}
\end{figure}

\begin{figure}[tb]
\begin{center}
\includegraphics[width=0.5\textwidth]{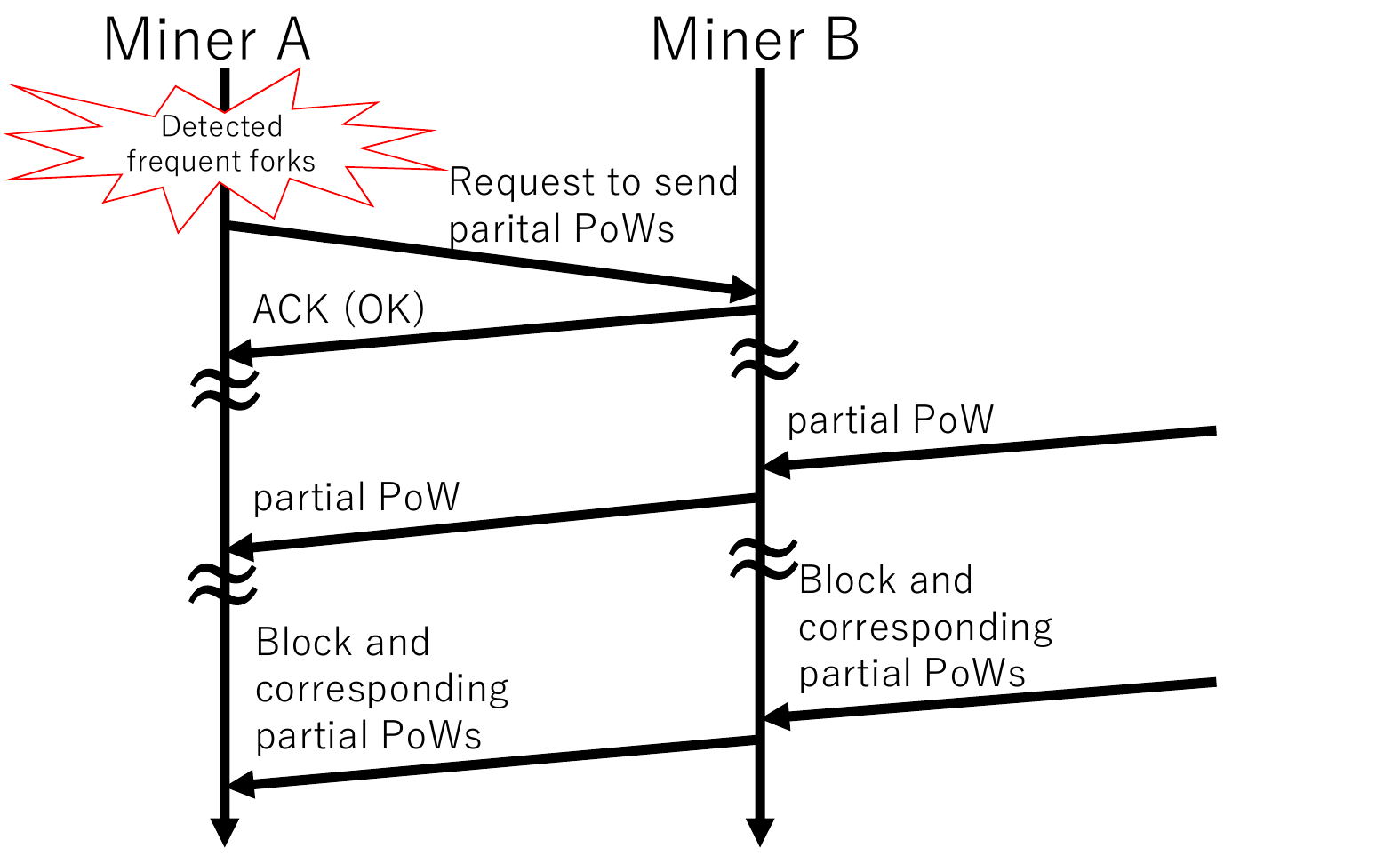}
\end{center}
\caption{Flow of sharing partial PoWs. When a miner has detected frequent forks, it sends the request message to send partial PoWs to its neighbor nodes. Miners do not have to send the block corresponding to a partial PoW. This is because miners can verify that the partial PoW has valid proof of work only with itself.}
\label{flow}
\end{figure}

\textbf{Embedding shared partial PoWs and mining.} The proposed method requires shared partial PoW to be embedded into blocks (Figure \ref{structure}). Using the OP\_RETURN\cite{OPRETURN} opcode is a known method for embedding shared partial PoW into a transaction. For example, an output can be reserved in the Coinbase transaction specifically for embedding the shared partial PoW.

We discuss the impact of our proposed method on the transaction-processing capacity. The size of the shared partial PoW depends on the cryptographic hash function that is used. For instance, 32 bytes are required when using SHA-256\cite{nistfips180-4}. When using the proposed method, no additional information must be inserted into the block. Considering that the current Bitcoin block size is 1 MiB, our proposed method has a minimal impact on the transaction-processing capacity.

\textbf{Committing to the difficulty adjuster.} The n-variable method method allows miners to choose the difficulty adjuster $n$. This stabilizes the supply of partial PoW regardless of the number of miners who use the method.

Each miner sets the difficulty adjuster $n$ in response to the amount of partial PoW currently published during mining. In other words, each miner commits to the difficulty adjuster $n$. The way of commitment is very simple. We use the timestamp of a block header. This timestamp is Unix time and proceeds per second. The determination of $n$ relies on interpreting the initial bits of this timestamp (Figure \ref{structure}). For example, if the first 2 bits are 00, $n = 50$. If the first 2 bits are 01, $n = 100$. If the first 2 bits are 10, $n = 150$. If the first 2 bits are 11, $n = 200$. 

Each miner sets $n = 50$ when the supply of partial PoW is sufficient, and vice versa. We consider the supply of partial PoW to be sufficient when there are more than 25 partial PoWs per block on average, which corresponds to a case where the proportion of hashrate of the attacker is $0.5$ and all honest miners adopt the proposed method and set $n = 50$.
For example, each miner first sets $n = 200$. If there are enough partial PoWs, say more than 100 partial PoWs per block on average, it resets $n = 150$. If there are  20 partial PoWs per block on average, it resets $n = 200$ again.

Miners compare the sum of the reciprocals of difficulty adjusters corresponding to each partial PoW in a chain tie with the n-variable method while they compare the number of shared partial PoW with the n-fixed method. If miners cannot find shared partial PoW in the block, they assign a value of 0 to the block.

The impact of the variable difficulty adjuster is negligible. This is because it only affects the initial bits of the timestamp and a few seconds. Miners cannot distinguish the effect of the method from that of the block propagation time. 

In addition to that, it is obvious that this n-variable method method has both forward compatibility and backward compatibility. Miners who do not adopt the method cannot distinguish blocks that are generated by themselves from blocks that are generated by miners who adopt the method. Similarly, miners who adopt the method cannot distinguish blocks that are generated by themselves from blocks that are generated by miners who do not adopt the method.

\textbf{Sharing partial PoWs.} Our method requires each miner to share a partial PoW. The way of sharing is arbitrary. For example, the same protocol used for propagating transactions can be used. 

We can share partial PoWs as needed (Figure \ref{flow}). The proposed method hardly changes the blockchain protocol, except for the tie-breaking rule. This allows for flexible responses such as using our proposed method only when the network is under attack and forks are frequent. For example, only when frequent forks are detected, each miner can send messages that require their neighbor nodes to send partial PoWs.
This as-needed application improves communication traffic per miner.

Next, let's consider the impact of our proposed method on the network bandwidth. When the proposed method is implemented in Bitcoin, a partial PoW is the block header, which is 80 bytes. That means the expected communication traffic is approximately $80n$ bytes per miner over 10 min. Considering that each miner shares a block's worth of transactions every 10 minutes, the sharing cost of partial PoW is small. This value is considerably smaller than the block size, and the increased bandwidth is acceptable. In addition, when a block is transmitted, it is necessary to send the set of partial PoWs corresponding to the shared partial PoW, which is approximately $80n$ bytes on average. The increased bandwidth can still be considered acceptable with an appropriate setting of $n$. Furthermore, the use of communication protocols that leverage set reconciliation, such as Graphene\cite{Graphene}, can further reduce bandwidth consumption.

\subsection{Partial Application}
We consider the effects in the scenario in which only some miners adopt the proposed method as this is the possibility. First, fewer partial PoWs are likely to be published. Increasing the difficulty adjuster $n$ as a countermeasure is conceivable, which makes it easier to generate partial PoW, thereby increasing the number.
In addition, an increase in the value of $\gamma$ is expected because of miners who do not adopt the proposed method. However, this effect is limited to miners who do not adopt the proposed method and does not negate the effectiveness of its partial application.

The proposed method is a tie-breaking rule. Hence, miners who adopt the proposed method do not generate stale blocks resulting from partial application. This property does not hold in the case of the application of fork choice rules such as GHOST \cite{GHOST} and weighted FRP\cite{PublishorPerish}.
\subsection{Exceeding of Block Propagation Time}
We examine the effects of the block propagation times exceeding $\Delta_B$. Possible effects include honest miners not receiving blocks within the acceptance window or attackers' partially shared PoW being considered sufficiently shared. However, as noted previously, a significant divergence occurs between the 100th percentile block propagation time and the hashrate-weighted average propagation time. Miners who hold the majority of the hashrate of the system tend to receive blocks far sooner than the 100th percentile propagation time. Therefore, even if miners with a smaller share of the hashrate experience delays in receiving blocks, such effects are likely to be minor and do not undermine the effectiveness of the proposed method. The same is true for partial PoW.

\subsection{Clock Drift}\label{drift}
Our proposed method considering clock drift allows us to replace the assumption that the impact of clock drift is negligible with the weaker assumption that there exists a widely known value $D$ such that the absolute value of each miner's clock drift is less than $D$. We examine how clock drift (the variation in how each miner's clock progresses) affects the method.

First, we define the clock drift $d_i$ of miner $i$ as
\begin{align}
 d_i = \frac{x_m - x_r}{x_r}
\end{align}
where $x_m$ represents the measured value of time progression, and $x_r$ is the actual value. We assume that there exists a widely known value $D$ such that $|d_i| < D$ holds for each miner $i$.

First, for the measured value $x_m$, the following holds:
\begin{align}
 x_m = x_r (1 + d_i) \leq x_r (1 + D) \label{r2m}
\end{align}
Similarly, for the actual value $x_r$,
\begin{align}
 x_r = \frac{x_m}{1 + d_i} \leq \frac{x_m}{1 - D} \label{m2r}
\end{align}
is also valid.

Next, we review the properties discussed in Section \ref{parameter}.

First, we revisit the parameter settings for the acceptance window during chain ties. When an attacker executes a Match against post-generated block, any honest miner can receive a block that is generated by another honest miner before the end of the acceptance window if the acceptance window $w$ is greater than $\Delta_B$. Each miner sets their acceptance window $w$ to be $\Delta_B (1 + D)$ to ensure that the acceptance window is in fact greater than $\Delta_B$, as derived from inequality \ref{r2m}.

Next, we reconsider the parameter settings for sufficiently shared partial PoW. If the time required for a valid partial PoW to be sufficiently shared is greater than $w + \Delta_B$, the partial PoW that is published by the attacker after an honest miner generates a block during a Match against post-generated block will not be regarded as sufficiently shared before the end of the acceptance window of any honest miner. Therefore, similar to the parameter settings for the acceptance window, it is appropriate to consider a partial PoW as sufficiently shared after $2\Delta_B (1 + D)$.

We review the parameter settings for the shared partial PoW. Each miner considers a valid partial PoW that has been measured for $2\Delta_B (1 + D)$ as sufficiently shared. The determination of the actual value, as derived from inequality \ref{m2r}, is equivalent to considering a partial PoW as sufficiently shared after $2\Delta_B (1 + D) / (1 - D)$. Therefore, a partial PoW should be included in the shared partial PoW after the actual duration $\Delta_P + 2\Delta_B (1 + D) / (1 - D)$. Similarly, a partial PoW should be included in the shared partial PoW after measuring $(\Delta_P + 2\Delta_B (1 + D) / (1 - D))(1 + D)$.

At this point, we revisit the evaluation of the effectiveness of the method. A partial PoW is included in the shared partial PoW after $(\Delta_P + 2\Delta_B (1 + D)/(1- D))(1 + D)$ as a measured value, which translates into $(\Delta_P + 2\Delta_B (1 + D)/ (1- D))(1 + D)/(1 - D)$ in actual terms. Therefore, similar to the derivation in Lemma \ref{lemma1}, the following holds:
\begin{lem} \label{lemma3}
    $\gamma$ satisfies the following inequality, on average:
    \begin{align*}
    \gamma \leq 1 -  \frac{n - 1}{n + \frac{\alpha}{1 - \alpha}} \exp (-\frac{ \Delta_B \frac{2 (1 + D)^2}{(1 - D)^2} + \Delta_P\frac{2}{1 - D}}{T}).
    \end{align*}
\end{lem}
Similarly, the following lemma also holds:
\begin{lem} \label{lemma4}
    $\gamma$ satisfies the following inequality, on average:
    \begin{align*}
    \gamma \leq 1 -  \frac{1}{2} \exp (-\frac{ \Delta_B \frac{2 (1 + D)^2}{(1 - D)^2} + \Delta_P\frac{2}{1 - D}}{T}).
    \end{align*}
\end{lem}

According to lemmas \ref{lemma3} and \ref{lemma4}, the following theorem holds in a system that adopts the proposed method considering clock drift:
\begin{thm} \label{third2}
$\gamma$ satisfies the following inequality, on average:
\begin{align*}
    \gamma \leq& \min( 1 - \frac{1}{2} \exp (-\frac{ \Delta_B \frac{2 (1 + D)^2}{(1 - D)^2} + \Delta_P\frac{2}{1 - D}}{T}) , \\&1 - \frac{n - 1}{n + \frac{\alpha}{1 - \alpha}} \exp (-\frac{ \Delta_B \frac{2 (1 + D)^2}{(1 - D)^2} + \Delta_P\frac{2}{1 - D}}{T})).
\end{align*}
\end{thm}

For example, assume $D = 0.1$. This indicates a clock drift of less than 6 min/h for each miner, which is considerably larger than the actual expected values \cite{InternalClockDriftEstimationinComputerClusters}. In this case, $ \Delta_B \frac{2 (1 + D)^2}{(1 - D)^2} + \Delta_P\frac{2}{1 - D}$ in Theorem \ref{third2} is approximately $2.99\Delta_B + 2.23 \Delta_P$. This implies that the effectiveness of the proposed method considering a clock drift of $D = 0.1$ with $\Delta_B = \Delta_P = 10$ is higher than that of the method that does not consider clock drift when $\Delta_B = \Delta_P = 20$ (as illustrated in Figure \ref{gamma2}).

\section{Conclusion} \label{conclusion}
We have proposed a novel last-generated rule based on partial PoW for PoW blockchain systems. Our method is designed to suppress intentional chain ties that are caused by attackers. When applied to Bitcoin, our method increases the threshold of the proportion of the attacker's hashrate that is necessary for selfish mining (SM) to $\alpha = 0.31479$, regardless of the block propagation abilities of the attacker. Whereas existing methods require extensive updates or challenging assumptions for PoW blockchain systems, our method can be easily applied to existing systems including Bitcoin. In addition, we have proposed extended selfish mining (ESM) to exploit last-generated rules and demonstrated the importance of parameter settings in last-generated rules.



\bibliography{hoge} 
\bibliographystyle{IEEEtran.bst}

\end{document}